\documentclass[11pt]{article}

\usepackage[margin=1in]{geometry}
\usepackage{amsmath,amssymb,amsfonts}
\usepackage{amsthm}
\usepackage{graphicx}
\usepackage{epsfig,multirow}

\def\d{\delta}
\def\r{\rho}
\def\e{\epsilon}

\def\a{\alpha}

\def\Ph{\boldsymbol{\Phi}}

\def\gm{\gamma}

\def\N{\mathbb{N}}
\def\A{\mathbf{A}}
\def\RR{\mathbb{R}}

\newcommand{\edges}{\mathcal E}
\newcommand{\support}{\mathcal S}
\newcommand{\bensemble}{\mathcal{B}}
\newcommand{\xensemble}{\mathbb{E}}

\newcommand{\obs}{\mathbf y}
\newcommand{\x}{\mathbf x}

\newcommand{\sol}{\widehat{\x}}
\newcommand{\bigO}{\mathcal O}

\newtheorem{theorem}{Theorem}[section]
\newtheorem{remark}{Remark}[section]
\newtheorem{corollary}{Corollary}[section]
\newtheorem{proposition}{Proposition}[section]
\newtheorem{definition}{Definition}[section]
\newtheorem{lemma}{Lemma}[section]

\title{On the construction of sparse matrices from expander graphs}

\author{Bubacarr Bah\thanks{African Institute for Mathematical Sciences (AIMS) South Africa, University of Stellenbosch ({\tt bubacarr@maims.ac.za})}, and Jared Tanner\thanks{Mathematics Institute, University of Oxford ({\tt tanner@maths.ox.ac.uk}).}}

\date{}

\begin{document}

\maketitle

\pagestyle{myheadings}
\thispagestyle{empty}
\markboth{BUBACARR BAH, AND JARED TANNER}{expander asymptotics}

\pagestyle{myheadings}


\begin{abstract}
We revisit the asymptotic analysis of probabilistic construction of adjacency matrices of expander graphs proposed in \cite{bah2013vanishingly}. With better bounds we derived a new reduced sample complexity for the number of nonzeros per column of these matrices, precisely $d = \bigO\left(\log_s(N/s) \right)$; as opposed to the standard $d = \bigO\left(\log(N/s) \right)$. This gives insights into why using small $d$ performed well in numerical experiments involving such matrices. Furthermore, we derive quantitative sampling theorems for our constructions which show our construction outperforming the existing state-of-the-art. We also used our results to compare performance of sparse recovery algorithms where these matrices are used for linear sketching.
\end{abstract}

\section{Introduction}\label{sec:intro}
Sparse binary matrices, say $\A\in \{0,1\}^{n\times N}$, with $n\ll N$ are widely used in applications including graph sketching \cite{ahn2012graph,gilbert2004compressing}, network tomography \cite{vardi1996network,castro2004network}, data streaming \cite{muthukrishnan2005data,indyk2007sketching}, breaking privacy of databases via aggregate queries \cite{dwork2007price}, compressed imaging of intensity patterns \cite{donoho1992maximum}, and more generally 
combinatorial compressed sensing \cite{donoho2006compressed, xu2007efficient,jafarpour2009efficient,berinde2009sequential,mendoza2017expander,mendoza2017robust}, linear sketching \cite{indyk2007sketching}, and group testing \cite{du2000combinatorial,gilbert2008group}. In all these areas we are interested in the case where $n\ll N$, in which case $\A$ is used as an efficient encoder of sparse signals $\x\in\RR^N$ with sparsity $s\ll n$, where they are known to preserve $\ell^1$ distance of sparse vectors \cite{berinde2008combining}.  Conditions that guarantee that a given encoder, $\A$, also referred to as a sensing matrix in compressed sensing, typically include the the nullspace, coherence, and the restricted isometry conditions, see \cite{foucart2013mathematical} and references there in. The goal is for $\A$ to satisfy one or more of these conditions with the minimum possible $n$, the number of measurements. For uniform guarantees over all $\A$, it has been established that $n$ has to be $\Omega\left(s^2\right)$, but that with high probability on the draw of random $\A$, $n$ can be $\bigO\left(s\log N/n \right)$ for $\A$ with entries drawn from a sub-gaussian distribution, see \cite{foucart2013mathematical} for a review of such results. Matrices with entries drawn from a Bernoulli distribution fall in the family of sub-gaussian but these are {\em dense} as opposed the the {\em sparse} binary matrices considered here. For 
computational advantages, such as faster application and smaller storage, it is advantageous to use sparse $\A$ in application 
\cite{berinde2008combining,bah2013vanishingly,mendoza2017expander}.

Herein we consider the $n$ achievable when $\A$ is an adjacency matrix of a expander graph \cite{berinde2008combining}, expander graph will be defined in the next section. Hence the construction of such matrices can be construed as either a linear algebra problem or equivalently a graph theory one (in this manuscript we will focus more on the linear algebra discourse). There has been significant research on expander graphs in pure mathematics and theoretical computer science, see  \cite{hoory2006expander} and references therein. Both deterministic and probabilistic constructions of expander graphs have been suggested. The best known deterministic constructions achieve $n = \bigO\left(s^{1+\a}\right)$ for $\a>0$ \cite{guruswami2009unbalanced}. One the other hand random constructions, first proven in \cite{bassalygo1973complexity}, achieve the optimal $n = \bigO\left(s\log\left(N/s\right) \right)$, precisely $n = \bigO(sd)$, with $d = \bigO\left(\log\left(N/s\right) \right)$, where $d$ is the {\em left degree} of the expander graph but also the number of ones in each column of $\A$, to be defined in the next section. However, to the best of our knowledge, it was \cite{bah2013vanishingly} that proposed a probabilistic construction that is not only optimal but also more suitable to making quantitative statements where such matrices are applied.

This work follows the probabilistic construction proposed in \cite{bah2013vanishingly} but with careful computation of the bounds, is able to achieve $n = \bigO\left(s\log\left(N/s\right)\right)$ with $d = \bigO\left(\frac{\log\left(N/s\right)}{\log s}\right)$. We retain the complexity of $n$ but got a smaller complexity for $d$, which is novel. Related results with a similar $d$ were derived in \cite{indyk2013model,bah2014model} but for structure sparse signals in the framework of model-based compressing sensing or sketching. In that framework, one has second order information about $\x$ beyond simple sparsity, which is first order information about $\x$. It is thus expected and established that it is possible to get a small $n$ and hence a smaller $d$. Arguably, such a small complexity for $d$ justifies in hindsight fixing $d$ to a small number in simulations with such $\A$ as in \cite{bah2013vanishingly,bah2014model,mendoza2017expander}, just to mention a few.

The results derive here are asymptotic, though finite dimensional bounds follow directly.  We focus on for what ratios of the problem dimensions $(s,n,N)$ does these results hold. There is almost a standard way of interrogating such a question, i.e. phase transitions, probably introduced to the compressed sensing literature by \cite{donoho2006thresholds}. In other words, we derive sampling theorems numerically depicted by phase transition plots about problem size spaces for which our construction holds. This is similar to what was done in \cite{bah2013vanishingly} but for comparison purposes we include phase transition plots from probabilistic constructions by \cite{buhrman2002bitvectors,berinde2009advances}. The plots show improvement over these earlier works. Furthermore, we show implications of our results for compressed sensing by using our results with the phase transition framework to compare the performance of selected combinatorial compressed sensing algorithms as is done in \cite{bah2013vanishingly,mendoza2017expander}.

The manuscript is organized as follows. Section \ref{sec:intro} gives the introduction; while Section \ref{sec:prelim} sets the notation and defines some useful terms. The main results are stated in Section \ref{sec:mresults} and the details of the construction is given in Section \ref{sec:construct}. This is followed by a discussion in Section \ref{sec:discuss} about our results, comparing them to existing results and using the results to compare the performance of some combinatorial compressed sensing algorithms. In Section \ref{sec:proof} we state the remaining proofs of theorems, lemmas, corollaries, and propositions used in this manuscript. After this section is the conclusion in Section \ref{sec:conclusion}. We include an appendix in Section  \ref{sec:appdx}, where we summarized key relevant materials from \cite{bah2013vanishingly}, and showed the derivation of some bounds used in the proofs.

\section{Preliminaries} \label{sec:prelim}

\subsection{Notation}
Scalars will be denoted by lowercase letters (e.g. $k$), vectors by lowercase boldface letters (e.g., ${\bf x}$), sets by uppercase calligraphic letters (e.g., $\mathcal{S}$) and matrices by uppercase boldface letters (e.g. ${\bf A}$).
The cardinality of a set $\mathcal{S}$ is denoted by $|\mathcal{S}|$ and $[N] := \{1, \ldots, N\}$.
Given $\mathcal{S} \subseteq [N]$, its complement is denoted by $\mathcal{S}^c := [N] \setminus \mathcal{S}$ and $\x_\mathcal{S}$ is the restriction of $\x \in \RR^N$ to $\mathcal{S}$, i.e.~ $(\x_\mathcal{S})_i = \x_i$ if $i \in \mathcal{S}$ and $0$ otherwise. For a matrix $\A$, the restriction of $\A$ to the columns indexed by $\support$ is denoted by $\A_\support$.
For a graph, $\Gamma(\support)$ denotes the set of {\em neighbors} of $\support$, that is the nodes that are connected to the nodes in $\support$, and $e_{ij} = (x_i,y_j)$ represents an edge connecting node $x_i$ to node $y_j$.
The $\ell_p$ norm of a vector ${\bf x} \in \RR^N$ is defined as $\|{\bf x}\|_p := \left ( \sum_{i=1}^N x_i^p \right )^{1/p}$.

\subsection{Definitions}
Below we give formal definitions that will be used in this manuscript. 
\begin{definition}[$\ell_p$-norm restricted isometry property]
 \label{def:rip}
A matrix $\A$ satisfies the $\ell_p$-norm restricted isometry property (RIP-p) of order $s$ and constant $\d_s <1$ if it satisfies the following inequality.
\begin{equation}
\label{eqn:rip}
\left(1-\d_s\right)\|\x\|_p^p \leq \|\A\x\|_p^p \leq \left(1+\d_s\right)\|\x\|_p^p, \quad \forall ~s\mbox{--sparse } \x.
\end{equation}
\end{definition}
\noindent The most popular case is RIP-2 and was first proposed in \cite{candes2006stable}. Typically when RIP is mentioned without qualification, it means $\mbox{RIP}_2$. In the discourse of this work though RIP-1 is the most relevant. The RIP says that $\A$ is a near-isometry and it is a sufficient condition to guarantee exact sparse recovery in the noiseless setting (i.e. $\obs = \A\x$); or recovery up to some error bound, also referred to as {\em optimality condition}, in the noisy setting (i.e. $\obs = \A\x + {\bf e}$, where ${\bf e}$ is the bounded noise vector). we define optimality condition more precisely below.
\begin{definition}[Optimality condition]
 \label{def:optimality}
Given $\obs = \A\x + {\bf e}$ and $\sol = \Delta\left(\A\x + {\bf e}\right)$ for a reconstruction algorithm $\Delta$, the optimal error guarantee is 
\begin{equation}
\label{eqn:optimality}
\|\sol - \x\|_p \leq C_1\sigma_s(\x)_q + C_2 \|{\bf x}\|_p\,,
\end{equation}
where $C_1,C_2 > 0$ depend only on the RIP constant (RIC), i.e. $\d_s$, and not the problem size, $1\leq q \leq p \leq 2$, and $\sigma_s(\x)_q$ denote the error of the best $s$-term approximation in the $\ell_q$-norm, that is
\begin{equation}
\label{eqn:beststerm}
\sigma_s(\x)_q  := \min_{s-\mbox{sparse } {\bf z}}\|{\bf z} - \x\|_q\,.
\end{equation}
\end{definition}
Equation \eqref{eqn:optimality} is also referred to as the $\ell_p/\ell_q$ optimality condition (or error guarantee). Ideally, we would like $\ell_2/\ell_2$, but the best provable is $\ell_2/\ell_1$ \cite{candes2006stable}, weaker than this is the $\ell_1/\ell_1$ \cite{berinde2008combining}, which is what is possible with the $\A$ considered in this work.

To aid translation between the terminology of graph theory and linear algebra we define the {\em set of neighbors} in both notation.

\begin{definition}[Definition 1.4 in \cite{bah2013vanishingly}]
 \label{def:neighbours}
Consider a bipartite graph $G=\left( [N],[n],\edges \right)$ where $\edges$ is the set of edges and $e_{ij}=(x_i,y_j)$ is the edge that connects vertex $x_i$ to vertex $y_j$. For a given set of left vertices $\support\subset [N]$ its set of neighbors is $\Gamma(\support) = \{y_j|x_i\in \support \mbox{ and } e_{ij}\in \edges\}$. In terms of the adjacency matrix, $\A$, of $G=\left( [N],[n],\edges \right)$ the set of neighbors of $\A_\support$ for $|\support|=s$, denoted by $A_s$, is the set of rows with at least one nonzero.
\end{definition}

\begin{definition}[Expander graph]
\label{def:llexpander}
Let $G=\left( [N],[n],\edges \right)$ be a left-regular bipartite graph with $N$ left vertexes, $n$ right vertexes, a set of edges $\edges$ and left degree $d$.
If, for any $\epsilon \in (0,1/2)$ and any $\support \subset [N]$ of size $|\support|\leq k$, we have that $|\Gamma(\support)| \geq (1-\epsilon)d|\support|$, then $G$ is referred to as a {\em $(s,d,\epsilon)$-expander graph}.
\end{definition}
\noindent The $\e$ is referred to as the expansion coefficient of the graph. A $(s,d,\e)$-expander graph, also called an {\em unbalanced} expander graph \cite{berinde2008combining} or a {\em lossless} expander graph \cite{capalbo2002randomness}, is a highly connected bipartite graph. We denote the ensemble of $n\times N$ binary matrices with $d$ ones per column by $\bensemble(N,n;d)$, or just $\bensemble$ to simplify notation. We also will denote the ensemble of $n\times N$ adjacency matrices of $(s,d,\e)$-expander graphs as $\xensemble(N,n;s,d,\e)$ or simply $\xensemble$.

\section{Results} \label{sec:mresults}
The main result of this work is formalized in Theorem \ref{thm:probconstruct}, which is an asymptotic result, where the dimensions grow while their ratios remain bounded. This is also referred to as the {\em propoational growth asymptotics} \cite{blanchard2011compressed,bah2010improved}.
\begin{theorem}
\label{thm:probconstruct}
Consider $\e\in(0,\frac{1}{2})$ and let $d,s,n,N\in \N$, a random draw of an $n\times N$ matrix $\A$ from $\bensemble$, i.e. for each column of $\A$ uniformly assign ones in $d$ out of $n$ positions, as $(s,n,N) \rightarrow \infty$ while $s/n \in (0,1)$ and $n/N \in (0,1)$, with probability approaching $1$ exponentially, the matrix $\A\in\xensemble$ with
\begin{equation}
\label{eqn:probconstruct}
d = \bigO\left(\frac{\log\left(N/s\right)}{\e\log s}\right), \quad \mbox{and} \quad n = \bigO\left(\frac{s\log\left(N/s\right)}{\e^2}\right).
\end{equation}
\end{theorem}
The proof of this theorem is found Section \ref{sec:pthm1}. It is worth emphasizing that the complexity of $d$ is novel and it is the main contribution of this work.

Furthermore, in the {\em proportional growth asymptotics}, i.e. as $(s,n,N) \rightarrow \infty$ while $s/n \rightarrow \r$ and $n/N \rightarrow \d$ with $\r,\d \in (0,1)$, for completeness, we derived a {\em phase transition} function (curve) in $\d\r$-space below which Theorem \ref{thm:probconstruct} is satisfied with high probability and the reverse is true. This is formalized in the following lemma. 

\begin{lemma}
\label{lem:prob_expander_existence}
Fix $\e\in(0,\frac{1}{2})$ and let $d,s,n,N\in \N$, as $(s,n,N) \rightarrow \infty$ while $s/n \rightarrow \r \in (0,1)$ and $n/N \rightarrow \d \in (0,1)$ then for $\r < (1-\gm)\r_{BT}(\d;d,\e)$ and $\gm>0$, a random draw of $\A$ from $\bensemble$ implies $\A\in\xensemble$ with probability approaching $1$ exponentially.
\end{lemma}
The proof of this lemma is given in Section \ref{sec:plem1}. The phase transition function $\r_{BT}(\d;d,\e)$ turned out to be significantly higher that those derived from existing probabilistic constructions, hence our results are significant improvement over earlier works. This will be graphically demonstrated with some numerical simulations in Section \ref{sec:discuss}.

\section{Construction} \label{sec:construct}
The standard probabilistic construction is for each column of $\A$ to uniformly assign ones in $d$ out of $n$ positions; while the standard approach to derive the probability bounds is to randomly selected $s$ columns of $\A$ indexed by $\support$ and compute the probability that $|A_s| < (1-\e)ds$, then do a {\em union bound} over all sets $\support$ of size $s$. Our work in \cite{bah2013vanishingly} computed smaller bounds than previous works based on a {\em dyadic splitting} of $\support$ and derived the following bound. We changed the notation and format of Theorem 1.6 in \cite{bah2013vanishingly} slightly to be consistent with the notation and format in this manuscript.

\begin{theorem}[Theorem 1.6, \cite{bah2013vanishingly}]
\label{thm:prob_expansion_bound}
Consider $d,s,n,N \in \N$, fix $\support$ with $|\support|\leq s$, let an $n\times N$ matrix $\A$ be drawn from $\bensemble$, then
\begin{equation}
\label{eqn:prob_expansion_bound1}
 \hbox{Prob}\left(\left|A_s\right| \leq a_s\right) < p_n(s,d) \cdot e^{\left[n\cdot\Psi_n\left(a_s,\ldots,a_1\right)\right]}
\end{equation}
with $a_1:=d$, and the functions defined as
\begin{align}
\label{eqn:pn}
p_n(s,d) & = \frac{2}{25\sqrt{2\pi s^3d^3}}, \quad \mbox{and}\\
\label{eqn:bigpsin}
 \Psi_n\left(a_s,\ldots,a_1\right) & = \frac{1}{n} \left[3s\log(5d) + \sum_{i\in\Omega} \frac{s}{2i} \psi_i \right], \quad \mbox{for} \quad \Omega = \{2^j\}_{j=0}^{\log_2(s) - 1}\,,
\end{align}
where 
\begin{equation}
\label{eqn:smallpsi}
\psi_i = \left(n-a_i\right) \cdot \mathcal{H}\left(\frac{a_{2i}-a_i}{n-a_i}\right) + a_i\cdot \mathcal{H}\left(\frac{a_{2i}-a_i}{a_i}\right) - n\cdot \mathcal{H}\left(\frac{a_i}{n}\right)\,,
\end{equation} 
and $\mathcal{H}(\cdot)$ is the Shannon entropy in base $e$ logarithms, and the index set $\Omega = \{2^j\}_{j=0}^{\log_2(s) - 1}$. 
\begin{itemize}
\item[a)] If no restriction is imposed on $a_s$ then the $a_i$ for $i>1$ take on their expected value $\hat{a}_{i}$ given by 
\begin{equation}
\label{eqn:restrictedsizes}
 \hat{a}_{2i} = \hat{a}_{i}\left(2 - \frac{\hat{a}_{i}}{n}\right), \quad \mbox{for} \quad i\in\{2^j\}_{j=0}^{\log_2(s) - 1}\,. 
\end{equation}
\item[b)]If $a_{s}$ is restricted to be less than $\hat{a}_{s}$, then the $a_i$ for $i>1$ are the unique solutions to the following polynomial system 
\begin{equation}
\label{eqn:unrestrictedsizes}
  a_{2i}^3 - 2a_ia_{2i}^2 + 2a_i^2a_{2i} - a_i^2a_{4i} = 0, \quad \mbox{for} \quad i\in\{2^j\}_{j=0}^{\log_2(s) - 2}\,,
\end{equation}
with $a_{2i}\geq a_i$ for each $i$.
\end{itemize}
\end{theorem}
In this work, based on the same approach as in \cite{bah2013vanishingly}, we derive new expressions for the $p_n(s,d)$ and $\Psi_n\left(a_s,\ldots,a_1\right)$ in Theorem \ref{thm:prob_expansion_bound}, i.e. \eqref{eqn:pn} and \eqref{eqn:bigpsin} respectively, and provide a simpler bound for the improved expression of $\Psi_n\left(a_s,\ldots,a_1\right)$. 

\begin{lemma}
\label{lem:prob_expansion_bound_new}
Theorem \ref{thm:prob_expansion_bound} holds with the functions
\begin{align}
\label{eqn:pn_new}
p_n(s,d) & = 2^{-3}s^{9/2}e^{1/4}\,,\\
\label{eqn:bigpsin_new}
 \Psi_n\left(a_s,\ldots,a_1\right) & = \frac{1}{n} \left[\frac{3\log 2}{2}\log_2^2 s + \left(\log_2 s - \frac{3}{2}\right)\log a_s + \sum_{i\in\Omega} \frac{s}{2i} \psi_i \right], ~~ \mbox{for} ~~ \Omega = \{2^j\}_{j=0}^{\log_2(s) - 1}\,.
\end{align}
\end{lemma}

The proof of the lemma is given in Section \ref{sec:plem2}. Asymptotically, the argument of the exponential term in the bound of the probability in \eqref{eqn:prob_expansion_bound1} of Theorem \ref{thm:prob_expansion_bound}, i.e. $\Psi_n\left(a_s,\ldots,a_1\right)$ in \eqref{eqn:bigpsin_new} is more important than the polynomial $p_n(s,d)$ in \eqref{eqn:pn_new} since the exponential factor will dominate the polynomial factor. The significance of the lemma is that $\Psi_n\left(a_s,\ldots,a_1\right)$ in \eqref{eqn:bigpsin_new} is smaller than $\Psi_n\left(a_s,\ldots,a_1\right)$ in \eqref{eqn:bigpsin} since $\frac{3\log 2}{2}\log_2^2 s + \left(\log_2 s - \frac{3}{2}\right)\log a_s$ in \eqref{eqn:bigpsin_new} is asymptotically smaller than $3s\log(5d)$ in \eqref{eqn:bigpsin}, because the former is $\bigO(\mbox{polylog} s)$ while the latter is $\bigO(s)$, since we consider $a_s = \bigO(s)$.

Recall that we are interested in computing $\hbox{Prob}\left(\left|A_s\right| \leq a_s\right)$ when $a_s = (1-\e)ds$. This means having to solve the polynomial equation \eqref{eqn:unrestrictedsizes} to compute as small a bound of $\Psi_n\left((1-\e)ds,\ldots, d \right)$ as possible. We derive an asymptotic solution to  \eqref{eqn:unrestrictedsizes} for $a_s = (1-\e)ds$ and use that solution to get the following bounds. 

\begin{theorem}
\label{thm:prob_expansion_bound_new}
Consider $d,s,n,N \in \N$, fix $\support$ with $|\support|\leq s$, for $\eta>0$, $\beta \geq 1$, and $\e\in\left(0,\frac{1}{2}\right)$, let an $n\times N$ matrix $\A$ be drawn from $\bensemble$, then
\begin{equation}
\label{eqn:prob_expansion_bound_new}
 \hbox{Prob}\left(\left|A_s\right| \leq (1-\e)ds\right) < p_n(s,d,\e) \cdot \exp{\left[n\cdot\Psi_n\left(s,d,\e \right)\right]}
\end{equation}
where 
\begin{align}
\label{eqn:pn_new2}
p_n(s,d,\e) & = \frac{\sqrt[4]{e} \cdot s^{\log_2(1-\e) + 3}}{\sqrt{2^{6}(1-\e)^3d^3}}\,,\\
\label{eqn:bigpsin_new2}
  \Psi_n\left(s,d,\e \right) &\leq -\frac{1}{2n}\left[\eta\beta^{-1}(\beta-1)(1-\e)ds\log_2(s/2) - (5\log 2)\log_2^2 s - 2\log d \log_2 s\right]\,.
\end{align}
\end{theorem}
The proof of this theorem is also found in Section \ref{sec:pthm3}. Since Theorem \ref{thm:prob_expansion_bound_new} holds for a fixed $\support$ of size at most $s$, if we want this to hold for all $\support$ of size at most $s$, we do a union bound over all $\support$ of size at most $s$. This leads to the following probability bound.

\begin{theorem}
\label{thm:prob_expansion_bound_new2}
Consider $d,s,n,N \in \N$, and all $\support$ with $|\support|\leq s$, for $\tau>0$, and $\e\in\left(0,\frac{1}{2}\right)$, let an $n\times N$ matrix $\A$ be drawn from $\bensemble$, then
\begin{equation}
\label{eqn:prob_expansion_bound_new2}
 \hbox{Prob}\left(\left|A_s\right| \leq (1-\e)ds\right) < p_N(s,d,\e) \cdot \exp{\left[N\cdot\Psi_N\left(s,d,\e \right)\right]}\,,
\end{equation}
where 
\begin{align}
\label{eqn:pN}
p_{N}(s,d,\e) & = \frac{5\cdot \sqrt[4]{e} \cdot s^{\log_2(1-\e) + 5/2}}{\sqrt{2^{10}(1-\e)^3d^3
\pi(1-s/N)}} \,,\\
\label{eqn:bigpsiN}
  \Psi_N\left(s,d,\e \right) &\leq -\frac{s}{N}\log\left(\frac{s}{N}\right) + \frac{s}{N} - \frac{\tau(1-\e)d}{2\log 2}\frac{s}{N}\log\left(\frac{s}{2}\right) + o(N)\,.
\end{align}
\end{theorem}

\begin{proof}
Applying the union bound over all $\support$ of size at most $s$ to \eqref{eqn:prob_expansion_bound_new} leads to the following. 
\begin{equation}
\label{eqn:unionbnd}
 \hbox{Prob}\left(\left|A_s\right| \leq (1-\e)ds\right) < \binom{N}{s} p_n(s,d,\e) \exp{\left[n\cdot\Psi_n\left(s,d,\e \right)\right]}.
\end{equation}
Then we used the upper bound of  \eqref{eqn:stirling} to bound the combinatorial term $\binom{N}{s}$ in \eqref{eqn:unionbnd}.
After some algebraic manipulations, we separated the the polynomial term, given in \eqref{eqn:pN}, from the exponential terms whose exponent is  
\begin{equation}
 \label{eqn:bigpsiN2}
 \Psi_N\left(s,d,\e \right) := \mathcal{H}\left(\frac{s}{N}\right) + \frac{n}{N}\Psi_n\left(s,d,\e \right).
 \end{equation} 
 We upper bound $\Psi_N\left(s,d,\e \right)$ in \eqref{eqn:bigpsiN} by upper bounding $\mathcal{H}\left(\frac{s}{N}\right)$ with $-\frac{s}{N}\log\left(\frac{s}{N}\right) + \frac{s}{N}$ and the upper bound of $\Psi_n\left(s,d,\e \right)$ in \eqref{eqn:bigpsin_new2}. The $o(N)$ decays to zeros with $N$ and its a result of dividing the polylogarithmic terms of $s$ in \eqref{eqn:bigpsin_new2}, and $\tau = \eta\beta^{-1}(\beta-1)$ in \eqref{eqn:bigpsin_new2}. This concludes the proof.
\end{proof}

The next corollary easily follows from Theorem \ref{thm:prob_expansion_bound_new2} and it is equivalent to Theorem \ref{thm:probconstruct}. Its statement is that if the conditions therein holds, then the probability that the cardinality of the set of neighbors of any $\support$ with $|\support|\leq s$ is less than $(1-\e)ds$ goes to zero as dimensions of $\A$ grows. On the other hand, the probability that the cardinality of the set of neighbors of any $\support$ with $|\support|\leq s$ is greater than $(1-\e)ds$ goes to one as dimensions of $\A$ grows. Implying that $\A$ is the adjacency matrix of a $(s,d,\e)$-expander graph.

\begin{corollary}
\label{cor:expander_asymptotics}
Given $d,s,n,N \in \N$, and $\e\in\left(0,\frac{1}{2}\right)$, for $d\geq \frac{c_d\log\left(N/s\right)}{\e\log s}$ and $n\geq \frac{c_ns\log\left(N/s\right)}{\e^2}$, with $c_d,c_n>0$. Let an $n\times N$ matrix $\A$ be drawn from $\bensemble$, in the proportional growth asymptotics 
\begin{equation}
\label{eqn:prob_expansion_bound_new3}
 \hbox{Prob}\left(\left|A_s\right| \leq (1-\e)ds\right) \rightarrow 0.
\end{equation}
\end{corollary} 

\begin{proof}
It suffice to focus on the exponent of \eqref{eqn:prob_expansion_bound_new2}, more precisely on the bound of $\Psi_N\left(s,d,\e \right)$ in \eqref{eqn:bigpsiN}, i.e.
\begin{equation}
 \label{eqn:bigpsiNbound}
 -\frac{s}{N}\log\left(\frac{s}{N}\right) + \frac{s}{N} - \frac{\tau(1-\e)d}{2\log 2}\frac{s}{N}\log\left(\frac{s}{2}\right) + o(N)\,.
  \end{equation} 
We can ignore the $o(N)$ term as this goes to zero as $N$ grows, and show that the remaining sum is negative. The remaining sum is
\begin{equation}
\label{eqn:complexity_condition}
-\frac{s}{N}\log\left(\frac{s}{N}\right) + \frac{s}{N} - \frac{\tau(1-\e)d}{2\log 2}\frac{s}{N}\log\left(\frac{s}{2}\right) = \frac{s}{N}\left[-\log\left(\frac{s}{N}\right) + 1 - \frac{\tau(1-\e)d}{2\log 2}\log\left(\frac{s}{2}\right)\right]\,.
\end{equation}
Hence, we can further focus on the sum in the square brackets, and find conditions on $d$ that will make it negative. We require
\begin{align}
\label{eqn:complexity_condition_d}
-\log\left(\frac{s}{N}\right) + 1 - \frac{\tau(1-\e)d}{2\log 2}\log\left(\frac{s}{2}\right) < 0, \quad & \Rightarrow \quad d > \frac{2\log2 \left(\log\left(N/s\right) + 1\right)}{\tau(1-\e)\log(s/2)} \\
\label{eqn:complexity_condition_d2}
& \Leftrightarrow \quad d\geq \frac{c_d\log(N/s)}{\e\log s}, \quad \exists ~ c_d > 0\,.
\end{align}
Recall $\tau = \eta\beta^{-1}(\beta-1)$ and $\beta$ is a function of $\e$, with $\beta(\e) \approx 1+\e$. Therefore, $\tau$ is a function of $\e$, and $\tau(\e) = \bigO(\e)$, hence there exists a $c_d > 0$ for \eqref{eqn:complexity_condition_d2} to hold. 

With regards to the complexity of $n$, we go back to the right hand side (RHS) of \eqref{eqn:complexity_condition_d} and we substitute $\frac{C_d\log\left(N/s\right)}{\e\log\left(s/2\right)}$ with $C_d>0$ for $d$ in the RHS of \eqref{eqn:complexity_condition_d} to get the following.
\begin{align}
\label{eqn:complexity_condition_n}
-\log\left(\frac{s}{N}\right) + 1 - \frac{C_d\tau(1-\e)\log\left(N/s\right)}{2\e \log 2} < 0\,.
\end{align}
Now we assume $n = \frac{C_ns\log\left(N/s\right)}{\e^2}$ with $C_n>0$ for $n$ and substitute this in \eqref{eqn:complexity_condition_n} to get the following. 
\begin{align}
\label{eqn:complexity_condition_n2}
-\log\left(\frac{s}{N}\right) + 1 - \frac{C_d\tau(1-\e)\e n}{2C_n s\log 2} < 0, \quad & \Rightarrow \quad n > \frac{2C_ns\log2 \left(\log\left(N/s\right) + 1\right)}{C_d\tau\e(1-\e)} \\
\label{eqn:complexity_condition_n3}
& \Leftrightarrow \quad n\geq \frac{c_ns\log\left(N/s\right)}{\e^2}\,, \quad \exists ~ c_n > 0\,.
\end{align}
Again since $\tau(\e) = \bigO(\e)$, hence there exists a $c_n > 0$ for \eqref{eqn:complexity_condition_n3} to hold.  The bound of $n$ in \eqref{eqn:complexity_condition_n3} agrees with our earlier assumption, thus concluding the proof.
\end{proof}

 \section{Discussion} \label{sec:discuss}
 \subsection{Comparison to other constructions}
 In addition to being the first probabilistic construction of adjacency matrices of expander graphs with such a small degree, quantitatively our results compares favorably to existing probabilistic constructions.We use the standard tool of phase transitions to compare our construction to the construction proposed in \cite{berinde2009advances} and those proposed in \cite{buhrman2002bitvectors}. The phase transition curve $\r_{BT}(\d;d,\e)$ we derived in Lemma \ref{lem:prob_expander_existence} is the $\r$ that solves the following equation.
 \begin{equation}
 \label{eqn:pt-bt}
 -\r\log\left(\d\r\right) + \r - \frac{\tau(1-\e)d\r\log(\d\r)}{2\log 2} - \frac{\tau\e(1-\e)d}{2c_n\log 2} + \frac{\tau(1-\e)d\r}{2} = 0\,, 
 \end{equation}
where $c_n>0$ is as in \eqref{eqn:complexity_condition_n3}. Equation \ref{eqn:pt-bt}  comes from taking the limit, in the proportional growth asymptotics, of the bound in \eqref{eqn:bigpsiN}, setting that to zero and simplifying. Similarly, for any $\support$ with $|\support|\leq s$, Berinde in \cite{berinde2009advances} derived the following bound on the set of neigbours of $\support$, i.e. $A_s$.
\begin{equation}
\label{eqn:prob_expansion_bound_BI}
\hbox{Prob}\left(\left|A_s\right| \leq (1-\e)ds\right) < \binom{N}{s} \binom{ds}{\e ds} \left(\frac{ds}{n}\right)^{\e ds}\,.
\end{equation}
We then express the bound in \eqref{eqn:prob_expansion_bound_BI} as the product of a polynomial term and an exponential term. A bound of the exponent is carefully derived as in the derivations above. We set the limit, in the proportional growth asymptotics, of this bound to zero and simplify to get the following.
 \begin{equation}
 \label{eqn:pt-bi}
 (\e d - 1)\log\r  - \log\d+ (1+\e d) - \e d\log\left(\e/d\right) = 0\,.
 \end{equation}
We refer to the $\r$ that solves \eqref{eqn:pt-bi} as the phase transition for the construction proposed by Berinde in \cite{berinde2009advances} and denote this $\r$ (the phase transition function) as $\r_{BI}(\d;d,\e)$. Another probabilistic construction was proposed by Burhman et al. in \cite{buhrman2002bitvectors}. In conforminty with the notation used in this manuscript their bound is equivalent to the following, also stated in a similar form by Indyk and Razenshteyn in \cite{indyk2013model}.
\begin{equation}
\label{eqn:prob_expansion_bound_BM}
\hbox{Prob}\left(\left|A_s\right| \leq (1-\e)ds\right) < \binom{N}{s} \left(\frac{\nu \e n}{ds}\right)^{-\e ds}\,,
\end{equation}
where $\nu > 0$. We again express the bound in \eqref{eqn:prob_expansion_bound_BM} as the product of a polynomial term and an exponential term. A bound of the exponent is carefully derived as in the derivations above. We set the limit, in the proportional growth asymptotics, of this bound to zero and simplify to get the following.
 \begin{equation}
 \label{eqn:pt-bm}
 (\e d - 1)\log\r  - \log\d+ 1 - \e d\log\left(\nu\e/d\right) = 0\,.
 \end{equation}
Similarly, we refer to the $\r$ that solves \eqref{eqn:pt-bm} as the phase transition for the construction proposed by Burhman et al. in \cite{buhrman2002bitvectors} and denote this $\r$ as $\r_{BM}(\d;d,\e)$. We compute numerical solutions to \eqref{eqn:pt-bt}, \eqref{eqn:pt-bi}, and  \eqref{eqn:pt-bm} to derive the phase transitions $\r_{BT}(\d;d,\e)$, $\r_{BI}(\d;d,\e)$, and $\r_{BM}(\d;d,\e)$ respectively. These are plotted in the {\em left panel} of Figure \ref{fig:pt-compare}. It is clear that our construction has a much higher phase transition than the others. Recall that the phase transition curves in these plots depict construction of adjacency matrices of $(s,d,\e)$-expanders with high probability for ratios of $s,n$ and $N$ (since $\r := s/n$, and $\d := n/N$) below the curve; and the failure to construct adjacency matrices of $(s,d,\e)$-expanders with high probability for ratios of $s,n$ and $N$ above the curve. Essentially, the larger the area under the curve the better.
\begin{remark}
It is easy to see that $\r_{BI}(\d;d,\e)$ is a special case of $\r_{BM}(\d;d,\e)$ since the two phase transitions will coincide, or equivalently \eqref{eqn:pt-bi} and  \eqref{eqn:pt-bm} will be the same, when $\nu = e^{-1}$. One could argue that Berinde's derivation in \cite{berinde2009advances} suffers from over counting.
\end{remark}

\begin{figure}[h]
\centering
\includegraphics[width=0.45\textwidth]{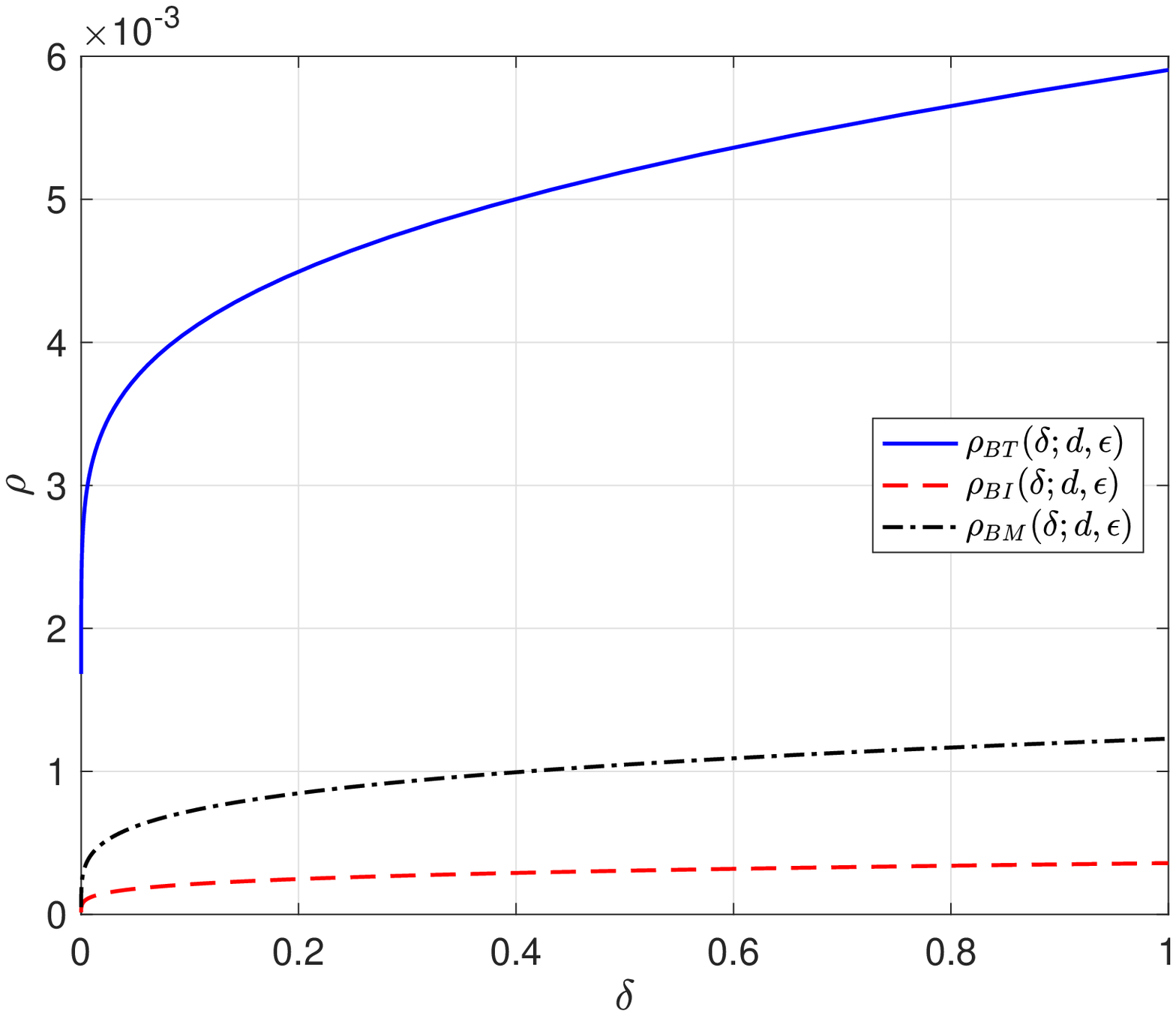} 
\includegraphics[width=0.45\textwidth]{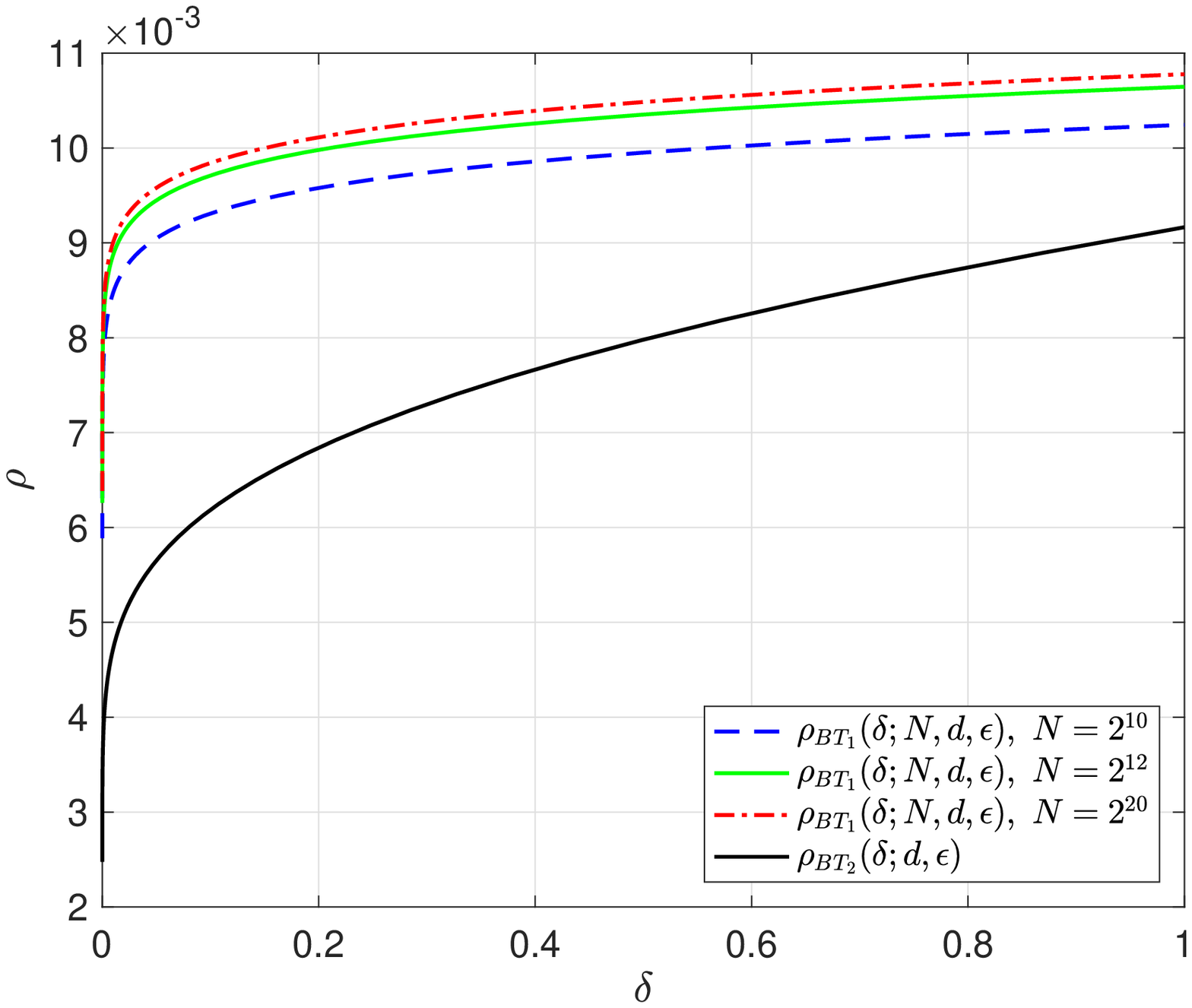} 
\caption{Phase transitions plots with fixed $d = 2^5$, $\e = 1/6$, and $\d\in\left[10^{-6},1\right]$ on a logarithmically spaced grid of 100 points. {\em Left panel:} A comparison of $\r_{BT}(\d;d,\e)$, $\r_{BI}(\d;d,\e)$, and $\r_{BM}(\d;d,\e)$, where $c_n = 2$. {\em Right panel:} A comparison of $\r_{BT}(\d;d,\e)$ denoted as $\r_{BT_2}(\d;d,\e)$ to our previous $\r_{BT}(\d;d,\e)$ denoted as $\r_{BT_1}(\d;d,\e)$ in \cite{bah2013vanishingly} with different values of $N$, (i.e. $2^{10}, 2^{12}$, and $2^{20}$) and $c_n = 2/3$.}
\label{fig:pt-compare}
\end{figure}

Given that this work is an improvement of our work in \cite{bah2013vanishingly} in terms of simplicity in computing $\r_{BT}(\d;d,\e)$, for completeness we compare our new phase transition $\r_{BT}(\d;d,\e)$ denoted as $\r_{BT_2}(\d;d,\e)$ to our previous $\r_{BT}(\d;d,\e)$ denoted as $\r_{BT_1}(\d;d,\e)$ in the {\em right panel} of Figure \ref{fig:pt-compare}. Each computation of $\r_{BT_1}(\d;d,\e)$ requires the specification of $N$, which is not needed in the computation of $\r_{BT_2}(\d;d,\e)$, hence the simplification. However, the simplification led to a lower phase transition as expected, which is confirmed by the plots in the {\em right panel} of Figure \ref{fig:pt-compare}. 

\begin{remark}
These simulations also inform us about the size of $c_n$. See from the plots of $\r_{BT}(\d;d,\e)$ and $\r_{BT_2}(\d;d,\e)$ that the smaller the value of $c_n$ the higher the phase transition but since $\r_{BT_2}(\d;d,\e)$ has to be a lower bound of $\r_{BT_1}(\d;d,\e)$, for values of $c_n$ much smaller than $2/3$, the lower bound will fail to hold. This informed the choice of $c_n=2$ in the plot of $\r_{BT}(\d;d,\e)$ in the {\em left panel} of Figure \ref{fig:pt-compare}.
\end{remark}
 
 \subsection{Implications for combinatorial compressed sensing}
When the sensing matrices are restricted to the sparse binary matrices considered in this manuscript, compressed sensing is usually referred to as {\em combinatorial compressed sensing} a term introduced in \cite{berinde2008combining} and used extensively in \cite{mendoza2017expander,mendoza2017robust}. In this setting, compressed sensing is more-or-less equivalent to linear sketching. The implications of our results on combinatorial compressed sensing are two-fold. One is on the $\ell_1$-norm RIP, we donate as RIP-1; while the second is in the comparison of performance of recovery algorithms for combinatorial compressed sensing.

\subsubsection{RIP-1}
As can be seen from \eqref{eqn:optimality}, the recovery errors in compressed sensing depend on the RIC, i.e. $\d_s$. The following lemma deduced from Theorem 1 of \cite{berinde2008combining} shows that a scaled $\A$ drawn from $\xensemble$ have RIP with $\d_s = 2\e$.
\begin{lemma}
\label{lem:rip1}
Consider $\e\in(0,\frac{1}{2})$ and let $\A$ be drawn from $\xensemble$, then $\Ph = \A/d$ satisfies the following RIP-1 condition
\begin{equation}
\label{eqn:rip1}
\left(1-2\e\right)\|\x\|_1 \leq \|\Ph\x\|_1 \leq \|\x\|_1, \quad \forall ~s\mbox{--sparse } \x.
\end{equation}
\end{lemma}
The interested reader is referred to the proof of Theorem 1 in \cite{berinde2008combining} for the proof of this lemma. Key to the holding of Lemma \ref{lem:rip1} is the existence of $(s,d,\e)$-expander graphs, hence one can draw corollaries from our results on this.

\begin{corollary}
\label{cor:probconstruct}
Consider $\e\in(0,\frac{1}{2})$ and let $d,s,n,N\in \N$. In the proportional growth asymptotics with a random draw of an $n\times N$ matrix $\A$ from $\bensemble$, the matrix $\Ph:=\A/d$ has RIP-1 with probability approaching $1$ exponentially, if
\begin{equation}
\label{eqn:probconstruct}
d = \bigO\left(\frac{\log\left(N/s\right)}{\e\log s}\right), \quad \mbox{and} \quad n = \bigO\left(\frac{s\log\left(N/s\right)}{\e^2}\right).
\end{equation}
\end{corollary}

\begin{proof}
Note that the upper bound of \eqref{eqn:rip1} holds trivially for any $\Ph = \A/d$ where $\A$ has $d$ ones per column, i.e. $\A\in\bensemble$. But for the lower bound of \eqref{eqn:rip1} to hold for any $\Ph = \A/d$, we need $\A$ to be an $(s,d,\e)$-expander matrix, i.e. $\A\in\xensemble$. Note that the event $\left|A_s\right| \geq (1-\e)ds$ is equal to the event $\|A_\support\x\|_1 \geq (1-2\e)d\|\x|_1$, which is equivalent to $\|\Ph_\support\x\|_1 \geq (1-2\e)\|\x|_1$, for a fixed $\support$, with $|\support|\leq s$. For $\A$ to be in $\xensemble$, we need expansion for all sets $\support$, with $|\support|\leq s$, i.e. $\A\in\xensemble$.  The key thing to remember is that
\begin{equation}
\label{eqn:probrip1}
\hbox{Prob}\left(\|\Ph\x\|_1 \geq (1-2\e)\|\x|_1\right) = \hbox{Prob}\left(\left|A_s\right| \geq (1-\e)ds\right)\,, \quad \mbox{for all} ~~ \{\support ~ : ~ |\support| \leq s\}\,.
\end{equation}
The probability in \eqref{eqn:probrip1} going to 1 exponentially in the proportional growth asymptotics, i.e. the existence of $\A\in\xensemble$ with parameters as given in \eqref{eqn:probconstruct}, is what is stated in Theorem \ref{thm:probconstruct}. Therefore, the rest of the proof follows from the proof of Theorem \ref{thm:probconstruct}, hence concluding the proof of the corollary.
\end{proof}

Notably, Lemma \ref{lem:rip1} holds with $\Ph$ having much smaller number of nonzeros per column due to our construction. More over, we can derive sampling theorems for which Lemma \ref{lem:rip1} holds as thus.
\begin{corollary}
\label{cor:prob_expander_existence}
Fix $\e\in(0,\frac{1}{2})$ and let $d,s,n,N\in \N$. In the proportional growth asymptotics, for any $\r < (1-\gm)\r_{BT}(\d;d,\e)$ and $\gm>0$, a random draw of $\A$ from $\bensemble$ implies $\Ph:=\A/d$ has RIP-1 with probability approaching $1$ exponentially.
\end{corollary}
\begin{proof}
The proof of this corollary follows from the proof of Corollary \ref{cor:probconstruct} above, and it is related to the proof of Lemma \ref{lem:prob_expander_existence} as the proof of Corollary \ref{cor:probconstruct} is to the proof of Theorem \ref{thm:probconstruct}. The details of the proof is thus skipped.
\end{proof}

\subsubsection{Performance of algorithms}
We wish to compare the performance of selected combinatorial compressed sensing algorithms in terms of the possible problem sizes $(s,n,N)$ that these algorithms can reconstruct sparse/compressible signals/vectors up to their respective error guarantees. The comparison is typically done in the framework of phase transitions, which depict a boundary curve where ratios of problems sizes above this curve are recovered with probability approaching 0 exponentially; while problems sizes below the curve are recovered with probability approaching 1 exponentially. The list of combinatorial compressed sensing algorithms includes Expander Matching Pursuit (EMP) \cite{indyk2008near}, Sparse Matching Pursuit \cite{berinde2008practical}, Sequential Sparse Matching Pursuit (SSMP) \cite{berinde2009sequential}, Left Degree Dependent Signal Recovery (LDDSR) \cite{xu2007efficient}, Expander Recovery (ER) \cite{jafarpour2009efficient}, Expander Iterative Hard-Thresholding (EIHT) \cite[Section 13.4]{foucart2013mathematical}, and Expander $\ell_0$-decoding (ELD) with both serial and parallel versions \cite{mendoza2017expander}. For reason similar to those used in \cite{bah2013vanishingly,mendoza2017expander}, we selected out of this list four of the algorithms: $(i)$ SSMP, $(ii)$ ER, $(ii)$ EIHT, $(iv)$ ELD. Descriptions of these algorithms is skipped here but the interested reader is referred to the original papers or their summarized details in \cite{bah2013vanishingly,mendoza2017expander}. We were also curious as to how $\ell_1$-minimization's performance compares to these selected combinatorial compressed sensing algorithms, since $\ell_1$-minimization ($\ell_1$-min) can be used to solve the combinatorial problem solved by these algorithms, see \cite[Theorem 3]{berinde2008combining}.

The phase transitions are based on conditions on the RIC of the sensing matrices used. Consequent to Lemma \ref{lem:rip1}, this becomes conditions on the expansion coefficient (i.e. $\e$) of the underlying $(s,d,\e)$-expander graphs of the sparse sensing matrices used. Where this condition on  $\e$ is not explicitly given it is easily deducible from the recovery guarantees given for each algorithms. The conditions are summarized in the table below. 

\begin{table}[h]
\centering
 \caption{Recovery conditions on $\e$}
 \begin{tabular}{|c|c|c||c|c|}
 \hline
 \multirow{2}{*}{\bf Algorithm}	& \multicolumn{2}{c||}{\bf Theoretical values} 	& \multicolumn{2}{c|}{\bf Computational values} \\
\cline{2-5}
 			& {\bf Condition}	& {\bf Sparsity} 		& {\bf Condition} 	& {\bf Sparsity}\\
  \hline
  \hline
  SSMP \cite{berinde2009sequential}			& $\e_k < 1/16$	& $k = (c+1)s, ~c>1$	& $\e_k = 1/16 - e$		& $k = \lceil (2+e)s \rceil, ~k = 3s$\\
  \hline
  ER	\cite{jafarpour2009efficient}			&	$\e_k < 1/4$	&	$k = 2s$		& $\e_k = 1/4 - e$		& $k = 2s$\\
  \hline
  EIHT \cite{foucart2013mathematical}		&	$\e_k < 1/12$	&	$k = 3s$		& $\e_k = 1/12 - e$		& $k = 3s$\\
  \hline
  ELD \cite{mendoza2017expander}			&	$\e_k \leq 1/4$	&	$k = s$		& $\e_k = 1/4$			& $k = s$	\\
  \hline
  $\ell_1$-min \cite{berinde2008combining}	&	$\e_k < 1/6$	&	$k = 2s$		& $\e_k = 1/6 - e$		& $k = 2s$\\
  \hline
 \end{tabular}
 \label{tab:e-condition}
\end{table}
The {\em theoretical values} are what will be found in the reference given in the table; while the {\em computational values} are what we used in our numerical experiments to compute the phase transition curves of the algorithms. The value for $e$ was set to be $10^{-15}$, to make the $\e_k$ as large as possible under the given condition. With these values we computed the phase transitions in Figure \ref{fig:pt-algo-compare}. 

\begin{figure}[h]
\centering
\includegraphics[width=0.45\textwidth]{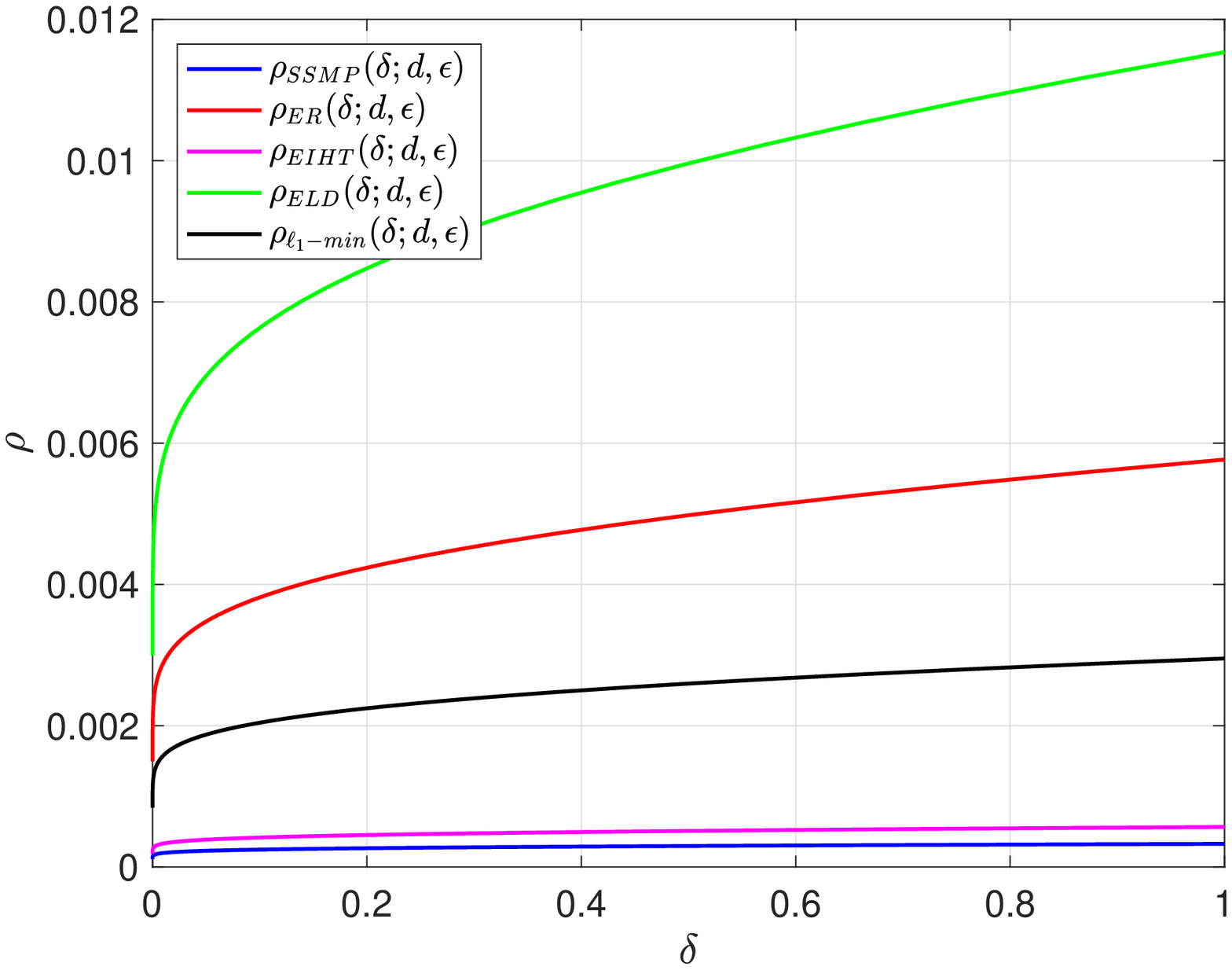} 
\includegraphics[width=0.45\textwidth]{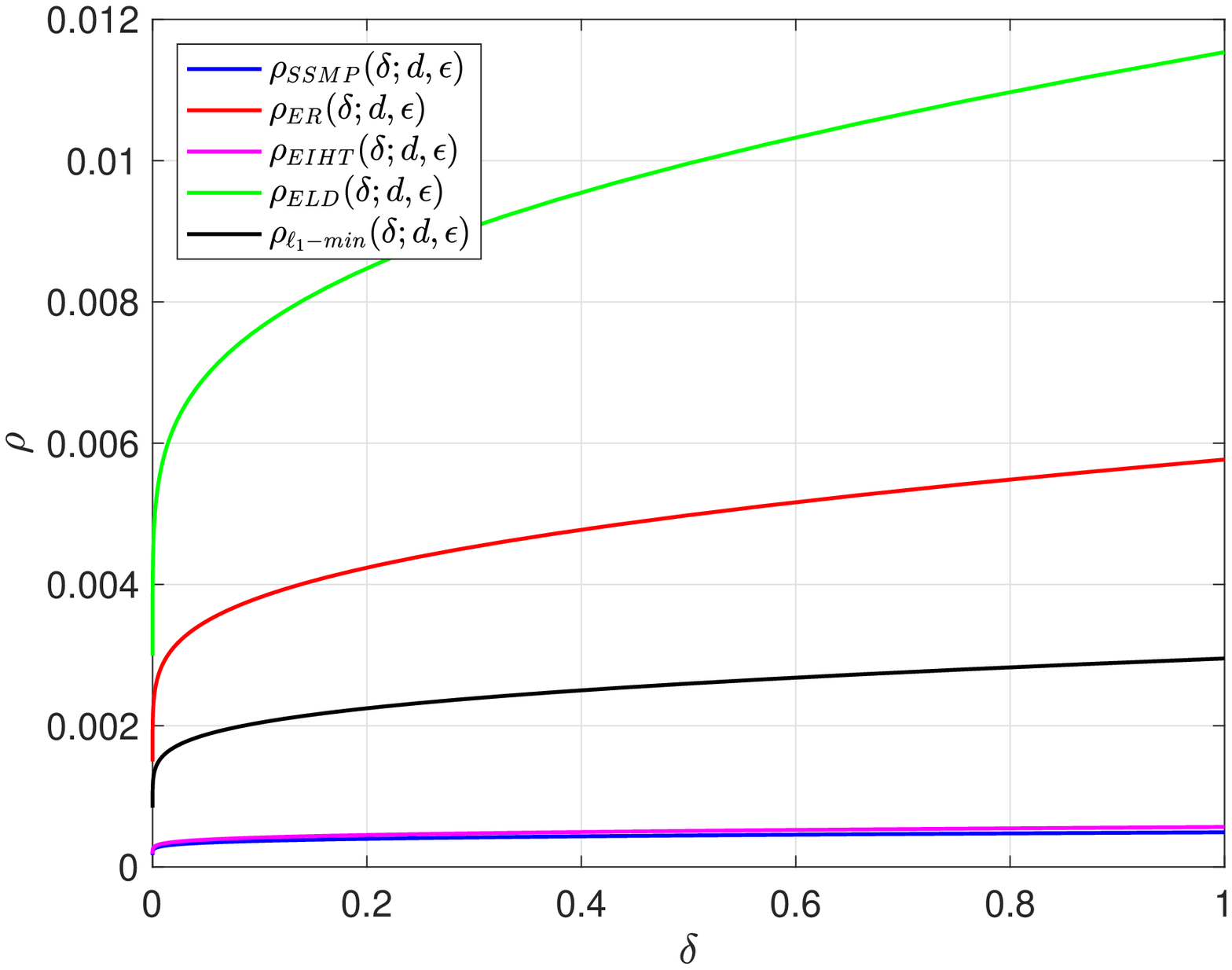} 
\caption{Phase transitions plots of algorithms with fixed $d = 2^5$, $\e$ as in the fourth column of Table  \ref{tab:e-condition} with $e = 10^{-15}$, $c_n = 2$ and $\d\in\left[10^{-6},1\right]$ on a logarithmically spaced grid of 100 points. {\em Left panel:} $k = 3s$ for $\r_{SSMP}(\d;d,\e)$. {\em Right panel:} $k = \lceil (2+e)s\rceil$ for $\r_{SSMP}(\d;d,\e)$.}
\label{fig:pt-algo-compare}
\end{figure}

The two figures are the same except for the different sparsity value used. The performance of the algorithms in this framework are thus ranked as follows: ELD, ER, $\ell_1$-min, EIHT, and SSMP.

\begin{remark}
\label{rem:comparison}
We point out that there are many way to compare performance of algorithms, this is just one way. For instance, we can compare runtime complexities or actual computational runtimes as in \cite{mendoza2017expander}; phase transitions of different probabilities, here the probability of recovery is 1 but this could be set to something else, like 1/2 in the simulations in \cite{mendoza2017expander}; one could also compare number of iterations and iteration cost as was also done in \cite{mendoza2017expander}.
\end{remark}

\section{Proofs} \label{sec:proof}\

\subsection{Theorem \ref{thm:probconstruct}} \label{sec:pthm1}
The proof of the theorem follows trivially from Corollary \ref{cor:expander_asymptotics}. Based on \eqref{eqn:prob_expansion_bound_new3} of Corollary \ref{cor:expander_asymptotics} we deduce that $\A\in\xensemble$ with probability approaching 1 exponentially with $d\geq \frac{c_d\log\left(N/s\right)}{\e\log s}$ and $n\geq \frac{c_ns\log\left(N/s\right)}{\e^2}$, with $c_d,c_n>0$, hence concluding the proof. \vspace{-0.7cm}
\begin{flushright}
$\square$
\end{flushright}

\subsection{Lemma \ref{lem:prob_expander_existence}} \label{sec:plem1}
The phase transition curve $\r_{BT}(\d;d,\e)$ is based the bound of the exponent of \eqref{eqn:prob_expansion_bound_new2}, which is
\begin{equation}
 \label{eqn:bigpsiNbound2}
 -\frac{s}{N}\log\left(\frac{s}{N}\right) + \frac{s}{N} - \frac{\tau(1-\e)d}{2\log 2}\frac{s}{N}\log\left(\frac{s}{2}\right) + o(N)\,.
  \end{equation} 
In the propotional growth asymptotics $(s,n,N) \rightarrow \infty$ while $s/n \rightarrow \r\in(0,1)$ and $n/N \rightarrow \d\in(0,1)$. This implies that $o(N) \rightarrow 0$ and  \eqref{eqn:bigpsiNbound2} becomes
 \begin{equation}
  \label{eqn:bigpsiNbound3}
 -\r\log\left(\d\r\right) + \r - \frac{\tau(1-\e)d\r\log(\d\r)}{2\log 2} - \frac{\tau\e(1-\e)d}{2c_n\log 2} + \frac{\tau(1-\e)d\r}{2}\,, 
 \end{equation}
where $c_n>0$ is as in \eqref{eqn:complexity_condition_n3}. If \eqref{eqn:bigpsiNbound3} is negative then as the problem size grows we have
\begin{equation}
\label{eqn:prob_expansion_bound_new4}
 \hbox{Prob}\left(\left|A_s\right| \leq (1-\e)ds\right) \rightarrow 0.
\end{equation}
Therefore, setting \eqref{eqn:bigpsiNbound3} to zero and solving for $\r$ gives us a critical $\r$ below which \eqref{eqn:bigpsiNbound3} is negative and positive above it. The critical $\r$ is the phase transition $\r$, i.e. $\r_{BT}(\d;d,\e)$, where below $\r_{BT}(\d;d,\e)$ is parameterized by the $\gamma$ in the lemma. This concludes the proof.
\vspace{-0.7cm}
\begin{flushright}
$\square$
\end{flushright}

\subsection{Lemma \ref{lem:prob_expansion_bound_new}} \label{sec:plem2}
By the dyadic splitting proposed in \cite{bah2013vanishingly}, we let $A_s=A_{\lceil\frac{s}{2}\rceil}^1\cup A_{\lfloor\frac{s}{2}\rfloor}^2$ such that $\left|A_s\right|=\left|A_{\lceil\frac{s}{2}\rceil}^1\cup A_{\lfloor\frac{s}{2}\rfloor}^2\right|$ and therefore
\begin{align}
\label{eqn:prob_1set_ssparseA}
\hbox{Prob}\left(\left|A_s\right| \leq a_s\right) & = \hbox{Prob}\left(\left|A_{\lceil\frac{s}{2}\rceil}^1\cup A_{\lfloor\frac{s}{2}\rfloor}^2\right| \leq a_s\right)\\
\label{eqn:prob_1set_ssparseB}
& = \sum_{l_s = d}^{a_s} \hbox{Prob}\left(\left|A_{\lceil\frac{s}{2}\rceil}^1\cup A_{\lfloor\frac{s}{2}\rfloor}^2\right| = l_s\right)
\end{align}
In \eqref{eqn:prob_1set_ssparseB} we sum over all possible events, i.e. all possible sizes of $l_s$. In line with the splitting technique, we simplify the probability to the product of the probabilities of the cardinalities of $\left|A_{\lceil\frac{s}{2}\rceil}^1\right|$ and $\left|A_{\lfloor\frac{s}{2}\rfloor}^2\right|$ and their intersection. Using the definition of $\hbox{P}_n(\cdot)$ in Lemma \ref{lem:intersect_prob} (Appendix \ref{sec:oldresults}), thus leads to the following.
\begin{multline}
\label{eqn:prob_1set_ssparseC}
\hbox{Prob}\left(\left|A_s\right| \leq a_s\right) = \sum_{l_s=2d}^{a_s} \sum_{l_{\lceil\frac{s}{2}\rceil}^1 = d}^{a_{\lceil\frac{s}{2}\rceil}} \sum_{l_{\lfloor\frac{s}{2}\rfloor}^2 = d}^{a_{\lfloor\frac{s}{2}\rfloor}} \hbox{P}_n\left(l_s,l_{\lceil\frac{s}{2}\rceil}^1,l_{\lfloor\frac{s}{2}\rfloor}^2\right) \\
\times \hbox{Prob}\left(\left|A_{\lceil\frac{s}{2}\rceil}^1\right| = l_{\lceil\frac{s}{2}\rceil}^1\right) \hbox{Prob}\left(\left|A_{\lfloor\frac{s}{2}\rfloor}^2\right| = l_{\lfloor\frac{s}{2}\rfloor}^2\right).
\end{multline}

In a slight abuse of notation we write $\displaystyle \mathop{\sum_{l^j}}_{j\in [m]}$ to denote applying the sum $m$ times. We also drop the limits of the summation indices henceforth. Now we use Lemma \ref{lem:num_size_split} in Appendix \ref{sec:oldresults} to simplify \eqref{eqn:prob_1set_ssparseC} as follows.
\begin{multline}
\label{eqn:prob_1set_ssparseD}
\mathop{\sum_{l^{j_1}_{Q_0}}}_{j_1\in [q_0]} \mathop{\sum_{l^{j_2}_{Q_1}}}_{j_2\in [q_1]} \mathop{\sum_{l^{j_3}_{R_1}}}_{j_3\in [r_1]} \hbox{P}_n\left(l^{j_1}_{Q_0},l^{2j_1-1}_{\lceil\frac{Q_0}{2}\rceil},l^{2j_1}_{\lfloor\frac{Q_0}{2}\rfloor}\right) \prod_{j_2 = 1}^{q_1}\hbox{Prob}\left(\left|A^{j_2}_{Q_1}\right| = l^{j_2}_{Q_1}\right)  \prod_{j_3=q_1+1}^{q_1+r_1} \hbox{Prob}\left(\left|A^{j_3}_{R_1}\right| = l^{j_3}_{R_1}\right).
\end{multline}

Now we proceed with the splitting - note \eqref{eqn:prob_1set_ssparseD} stopped only at the first level. At the next level, the second, we will have $q_2$ sets with $Q_2$ columns and $r_2$ sets with $R_2$ columns which leads to the following expression.
\begin{multline}
\label{eqn:prob_1set_ssparseE}
\mathop{\sum_{l^{j_1}_{Q_0}}}_{j_1\in [q_0]} \mathop{\sum_{l^{j_2}_{Q_1}}}_{j_2\in [q_1]} \mathop{\sum_{l^{j_3}_{R_1}}}_{j_3\in [r_1]} \hbox{P}_n\left(l^{j_1}_{Q_0},l^{2j_1-1}_{\lceil\frac{Q_0}{2}\rceil},l^{2j_1}_{\lfloor\frac{Q_0}{2}\rfloor}\right) \Bigg{[} \mathop{\sum_{l^{j_4}_{Q_2}}}_{j_4\in [q_2]} \mathop{\sum_{l^{j_5}_{R_2}}}_{j_5\in [r_2]}
\hbox{P}_n\left(l^{j_2}_{Q_1},l^{2j_2-1}_{\lceil\frac{Q_1}{2}\rceil},l^{2j_2}_{\lfloor\frac{Q_1}{2}\rfloor}\right) \\ \times \hbox{P}_n\left(l^{j_3}_{R_1},l^{2j_3-1}_{\lceil\frac{R_1}{2}\rceil},l^{2j_3}_{\lfloor\frac{R_1}{2}\rfloor}\right) \prod_{j_4 = 1}^{q_2}\hbox{Prob}\left(\left|A^{j_4}_{Q_2}\right| = l^{j_4}_{Q_2}\right) \prod_{j_5=q_2+1}^{q_2+r_2} \hbox{Prob}\left(\left|A^{j_5}_{R_2}\right| = l^{j_5}_{R_2}\right)\Bigg{]}.
\end{multline}

We continue this splitting of each instance of $\hbox{Prob}(\cdot)$ for $\lceil \log_2 s\rceil -1$ levels until reaching sets with single columns where, by construction, the probability that the single column has $d$ nonzeros is one. Note that at this point we drop the subscripts $j_i$, as they are no longer needed. This process gives a complicated product of nested sums of $\hbox{P}_n(\cdot)$ which we express as

\begin{multline}
\label{eqn:prob_1set_ssparseF}
\mathop{\sum_{l_{Q_0}}} \mathop{\sum_{l_{Q_1}}} \mathop{\sum_{l_{R_1}}}\hbox{P}_n\left(l_{Q_0},l_{\lceil\frac{Q_0}{2}\rceil},l_{\lfloor\frac{Q_0}{2}\rfloor}\right) \Bigg{[} \mathop{\sum_{l_{Q_2}}} \mathop{\sum_{l_{R_2}}}
\hbox{P}_n\left(l_{Q_1},l_{\lceil\frac{Q_1}{2}\rceil},l_{\lfloor\frac{Q_1}{2}\rfloor}\right) \hbox{P}_n\left(l_{R_1},l_{\lceil\frac{R_1}{2}\rceil},l_{\lfloor\frac{R_1}{2}\rfloor}\right) \Bigg{[} \cdots \\
\times \cdots \Bigg{[} \mathop{\sum_{l_{Q_{\lceil\log_2s\rceil-1}}}} \hbox{P}_n\left(l_{4}, l_{2}, l_{2}\right)  \hbox{P}_n\left(l_{3}, l_{2}, d\right)  \hbox{P}_n\left(l_{2}, d, d\right)\Bigg{]} \cdots \Bigg{]}.
\end{multline}

Using the expression for $\hbox{P}_n(\cdot)$ in \eqref{eqn:intersect_prob} of Lemma \ref{lem:intersect_prob} (Appendix \ref{sec:oldresults}) we bound \eqref{eqn:prob_1set_ssparseF} by bounding each $\hbox{P}_n(\cdot)$ as in \eqref{eqn:prob_1set_poly1} with a product of a polynomial, $\pi(\cdot)$, and an exponential with exponent $\psi_n(\cdot)$.

\begin{multline}
\label{eqn:prob_1set_ssparseG}
\mathop{\sum_{l_{Q_0}}} \mathop{\sum_{l_{Q_1}}} \mathop{\sum_{l_{R_1}}} \pi\left(l_{Q_0},l_{\lceil\frac{Q_0}{2}\rceil},l_{\lfloor\frac{Q_0}{2}\rfloor}\right) e^{ \psi_n\left(l_{Q_0},l_{\lceil\frac{Q_0}{2}\rceil},l_{\lfloor\frac{Q_0}{2}\rfloor}\right)}
\cdot\Bigg{[} \mathop{\sum_{l_{Q_2}}} \mathop{\sum_{l_{R_2}}}
\pi\left(l_{Q_1},l_{\lceil\frac{Q_1}{2}\rceil},l_{\lfloor\frac{Q_1}{2}\rfloor}\right) \times \\
e^{\psi_n\left(l_{Q_1},l_{\lceil\frac{Q_1}{2}\rceil},l_{\lfloor\frac{Q_1}{2}\rfloor}\right)} 
\pi\left(l_{R_1},l_{\lceil\frac{R_1}{2}\rceil},l_{\lfloor\frac{R_1}{2}\rfloor}\right) \cdot e^{\psi_n\left(l_{R_1},l_{\lceil\frac{R_1}{2}\rceil},l_{\lfloor\frac{R_1}{2}\rfloor}\right)} \Bigg{[} \cdots \times \Bigg{[} \mathop{\sum_{l_{Q_{\lceil\log_2s\rceil-1}}}}\pi\left(l_{4}, l_{2}, l_{2}\right) \\ 
\times e^{\psi_n\left(l_{4}, l_{2}, l_{2}\right)} \pi\left(l_{3}, l_{2}, d\right) e^{\psi_n\left(l_{3}, l_{2}, d\right)}  \pi\left(l_{2}, d, d\right) \cdot e^{\psi_n\left(l_{2}, d, d\right)}\Bigg{]} \cdots \Bigg{]}.
\end{multline}

Using Lemma \ref{lem:psi_n_behaviour} in Appendix \ref{sec:oldresults} we maximize the $\psi_n(\cdot)$ and hence the exponentials in \eqref{eqn:prob_1set_ssparseG}. We maximize each $\psi_n(\cdot)$ by choosing $l_{(\cdot)}$ to be $a_{(\cdot)}$. Then \eqref{eqn:prob_1set_ssparseG} will be upper bounded by the following.

\begin{multline}
\label{eqn:prob_1set_ssparseH}
\mathop{\sum_{l_{Q_0}}} \mathop{\sum_{l_{Q_1}}} \mathop{\sum_{l_{R_1}}} \pi\left(l_{Q_0},l_{\lceil\frac{Q_0}{2}\rceil},l_{\lfloor\frac{Q_0}{2}\rfloor}\right) e^{ \psi_n\left(a_{Q_0},a_{\lceil\frac{Q_0}{2}\rceil},a_{\lfloor\frac{Q_0}{2}\rfloor}\right)}
\cdot\Bigg{[} \mathop{\sum_{l_{Q_2}}} \mathop{\sum_{l_{R_2}}}
\pi\left(l_{Q_1},l_{\lceil\frac{Q_1}{2}\rceil},l_{\lfloor\frac{Q_1}{2}\rfloor}\right) \times \\
e^{\psi_n\left(a_{Q_1},a_{\lceil\frac{Q_1}{2}\rceil},a_{\lfloor\frac{Q_1}{2}\rfloor}\right)} 
\pi\left(l_{R_1},l_{\lceil\frac{R_1}{2}\rceil},l_{\lfloor\frac{R_1}{2}\rfloor}\right) \cdot e^{\psi_n\left(a_{R_1},a_{\lceil\frac{R_1}{2}\rceil},a_{\lfloor\frac{R_1}{2}\rfloor}\right)} \Bigg{[} \ldots \times \Bigg{[} \mathop{\sum_{l_{Q_{\lceil\log_2s\rceil-1}}}}\pi\left(l_{4}, l_{2}, l_{2}\right) \\ 
\times e^{\psi_n\left(a_{4}, a_{2}, a_{2}\right)} \pi\left(l_{3}, l_{2}, d\right) e^{\psi_n\left(a_{3}, a_{2}, d\right)}  \pi\left(l_{2}, d, d\right) \cdot e^{\psi_n\left(a_{2}, d, d\right)}\Bigg{]} \ldots \Bigg{]}.
\end{multline}

We then factor the product of exponentials. This product becomes an exponential where exponent is the summation of the $\psi_n(\cdot)$, we will denote this exponent as $\widetilde{\Psi}_n\left(a_s,\ldots,a_2,d\right)$. Then \eqref{eqn:prob_1set_ssparseH} simplifies to the following.

\begin{multline}
\label{eqn:prob_1set_ssparseI}
e^{\widetilde{\Psi}_n\left(a_s,\ldots,a_2,d\right)}\cdot\mathop{\sum_{l_{Q_0}}} \mathop{\sum_{l_{Q_1}}} \mathop{\sum_{l_{R_1}}} \pi\left(l_{Q_0},l_{\lceil\frac{Q_0}{2}\rceil},l_{\lfloor\frac{Q_0}{2}\rfloor}\right) 
\cdot\Bigg{[} \mathop{\sum_{l_{Q_2}}} \mathop{\sum_{l_{R_2}}}
\pi\left(l_{Q_1},l_{\lceil\frac{Q_1}{2}\rceil},l_{\lfloor\frac{Q_1}{2}\rfloor}\right) \times \\
\pi\left(l_{R_1},l_{\lceil\frac{R_1}{2}\rceil},l_{\lfloor\frac{R_1}{2}\rfloor}\right) \cdot \Bigg{[} \ldots \times \Bigg{[} \mathop{\sum_{l_{Q_{\lceil\log_2s\rceil-1}}}}\pi\left(l_{4}, l_{2}, l_{2}\right) \pi\left(l_{3}, l_{2}, d\right) \pi\left(l_{2}, d, d\right)\Bigg{]} \ldots \Bigg{]}.
\end{multline}

We denote the factor multiplying the exponential term by $\Pi\left(l_s,\ldots,l_2,d\right)$, therefore we have the following bound.
\begin{equation}
\label{eqn:prob_1set_ssparseJ}
\hbox{Prob}\left(\left|A_s\right| \leq a_s\right) \leq \Pi\left(l_s,\ldots,l_2,d\right) \cdot \exp\left(\widetilde{\Psi}_n\left(a_s,\ldots,a_2,d\right)\right)\,,
\end{equation}
where $\widetilde{\Psi}_n\left(a_s,\ldots,a_2,d\right)$ is exactly $\Psi_n\left(a_s,\ldots,a_2,d\right)$ given by (78) in \cite[Proof of Theorem 1.6]{bah2013vanishingly}. Consequently, we state the bound of $\widetilde{\Psi}_n\left(a_s,\ldots,a_2,d\right)$ and skip the proof, which is as thus.
\begin{equation}
\label{eqn:Psi_n_bound}
\widetilde{\Psi}_n\left(a_s,\ldots,a_2,d\right) \leq
\sum_{i\in\Omega} \frac{s}{2i} \psi_i, \quad \mbox{for} \quad \Omega = \{2^j\}_{j=0}^{\log_2(s) - 1}\,,
\end{equation}
where $\psi_i$ is given by \eqref{eqn:smallpsi}.
The upper bound of $\Pi\left(l_s,\ldots,l_2,d\right)$ is given by the following proposition.
\begin{proposition}
\label{pro:pi_bound}
Given $l_s\leq a_s, ~l_{\lceil s/2 \rceil} \leq a_{\lceil s/2 \rceil}, \cdots, l_2\leq a_2$, we have
\begin{equation}
\label{eqn:pi_bound}
\Pi\left(l_s,\ldots,l_2,d\right) \leq 2^{-3} \cdot e^{\frac{1}{4}} \cdot s^{\frac{3}{2}\log_2 s + \frac{9}{2}} \cdot \left(a_{s}\right)^{\log_2 s - \frac{3}{2}} \,.
\end{equation}
\end{proposition}
The proof of the proposition is found in Section \ref{sec:ppro1}. Taking log of right hand side of \eqref{eqn:pi_bound} and then exponentiating the results yields 
\begin{align}
\label{eqn:pi_bound2}
\Pi\left(l_s,\ldots,l_2,d\right) & \leq \exp\left[\frac{1}{4}  - 3\log 2 + \left(\frac{9}{2}+\frac{3}{2}\log_2 s\right)\log s + \left(\log_2 s - \frac{3}{2}\right)\log a_s\right]\\
\label{eqn:pi_bound3}
& = 2^{-3}s^{9/2}e^{1/4}\cdot\exp\left[\frac{3}{2}\log_2 s\log s + \left(\log_2 s - \frac{3}{2}\right)\log a_s\right]\\
\label{eqn:pi_bound4}
& = 2^{-3}s^{9/2}e^{1/4}\cdot\exp\left[\frac{3\log 2}{2}\log_2^2 s + \left(\log_2 s - \frac{3}{2}\right)\log a_s\right]\,.
\end{align}
Combining \eqref{eqn:pi_bound4} and \eqref{eqn:Psi_n_bound} gives the following bound for \eqref{eqn:prob_1set_ssparseJ}
\begin{equation}
\label{eqn:prob_1set_ssparseK}
\hbox{Prob}\left(\left|A_s\right| \leq a_s\right) \leq 2^{-3}s^{9/2}e^{1/4}\cdot \exp\left[\frac{3\log 2}{2}\log_2^2 s + \left(\log_2 s - \frac{3}{2}\right)\log a_s + \widetilde{\Psi}_n\left(a_s,\ldots,a_2,d\right)\right]\,.
\end{equation}
It follows therefore that $p_n(s,d) = 2^{-3}s^{9/2}e^{1/4}$ as in \eqref{eqn:pn_new} and the exponent in \eqref{eqn:prob_1set_ssparseK} is $n\cdot \Psi_n\left(a_s,\ldots,a_1\right)$, which implies
\eqref{eqn:bigpsin_new}. This concludes the proof of the lemma.
\vspace{-0.7cm}
\begin{flushright}
$\square$
\end{flushright}

\subsection{Proposition \ref{pro:pi_bound}} \label{sec:ppro1}
By definition, from \eqref{eqn:prob_1set_ssparseI}, we have 
\begin{multline}
\label{eqn:prob_1set_ssparseK2}
\Pi\left(l_s,\ldots,l_2,d\right) := \mathop{\sum_{l_{Q_0}}} \mathop{\sum_{l_{Q_1}}} \mathop{\sum_{l_{R_1}}} \pi\left(l_{Q_0},l_{\lceil\frac{Q_0}{2}\rceil},l_{\lfloor\frac{Q_0}{2}\rfloor}\right) 
\cdot\Bigg{[} \mathop{\sum_{l_{Q_2}}} \mathop{\sum_{l_{R_2}}}
\pi\left(l_{Q_1},l_{\lceil\frac{Q_1}{2}\rceil},l_{\lfloor\frac{Q_1}{2}\rfloor}\right) \times \\
\pi\left(l_{R_1},l_{\lceil\frac{R_1}{2}\rceil},l_{\lfloor\frac{R_1}{2}\rfloor}\right) \cdot \Bigg{[} \cdots \times \Bigg{[} \mathop{\sum_{l_{Q_{\lceil\log_2s\rceil-1}}}}\pi\left(l_{4}, l_{2}, l_{2}\right) \pi\left(l_{3}, l_{2}, d\right) \pi\left(l_{2}, d, d\right)\Bigg{]} \cdots \Bigg{]}.
\end{multline}

From \eqref{eqn:prob_1set_poly3d} we see that $\pi(\cdot)$ is maximized when all the three arguments are the same and using Corollary \ref{cor:pi_monotonicity} we take largest possible arguments that are equal in the range of the summation. Before we write out the resulting bound for $\Pi\left(l_s,\ldots,l_2,d\right)$, we simplify notation by denoting $\pi(x,x,x)$ as $\pi(x)$, and noting that $Q_{\lceil\log_2s\rceil-1} = 2$. Therefore, the bound becomes the following.
\begin{multline}
\label{eqn:prob_1set_ssparseL}
\Pi\left(l_s,\ldots,l_2,d\right) \leq \mathop{\sum_{l_{Q_0}}} \mathop{\sum_{l_{Q_1}}} \mathop{\sum_{l_{R_1}}} \pi\left(l_{Q_0}\right) \Bigg{[} \mathop{\sum_{l_{Q_2}}} \mathop{\sum_{l_{R_2}}}
\pi\left(l_{Q_1}\right) \\
\times \pi\left(l_{R_1}\right) \Bigg{[} \cdots \Bigg{[} \mathop{\sum_{l_{2}}}\pi\left(l_{4}\right) \pi\left(l_{3}\right) \pi\left(l_{2}\right)\Bigg{]} \cdots \Bigg{]}.
\end{multline}
Properly aligning the $\pi(\cdot)$ with their relevant summations simplifies the right hand side (RHS) of \eqref{eqn:prob_1set_ssparseL} to the following.
\begin{equation}
\label{eqn:prob_1set_ssparseM}
\left[\mathop{\sum_{l_{Q_0}}} \pi\left(l_{Q_0}\right) \right] \cdot \left[\mathop{\sum_{l_{Q_1}}} \pi\left(l_{Q_1}\right) \right] \cdot \left[\mathop{\sum_{l_{R_1}}} \pi\left(l_{R_1}\right)\right] \times \cdots \times
\left[\mathop{\sum_{l_{4}}} \pi\left(l_{4}\right)\right] \cdot
\left[ \mathop{\sum_{l_{3}}} \pi\left(l_{3}\right)\right] \cdot \left[\mathop{\sum_{l_{2}}} \pi\left(l_{2}\right)\right]\,.
\end{equation}
From \eqref{eqn:prob_1set_poly3d} we have
\begin{equation}
\label{eqn:psi_n_bound}
\pi(y) := \pi(y,y,y) = \left(\frac{5}{4} \right)^2 \sqrt{\frac{2\pi y(n-y)}{n}} < \left(\frac{5}{4} \right)^2 \sqrt{2\pi y}\,.
\end{equation}
We use the RHS of \eqref{eqn:psi_n_bound} to upper bound each term in \eqref{eqn:prob_1set_ssparseM}, leading to the following bound.
\begin{multline}
\label{eqn:prob_1set_ssparseN}
\left[\sqrt{2\pi}\left(\frac{5}{4} \right)^2 \mathop{\sum_{l_{Q_0}}} \sqrt{l_{Q_0}}\right] \cdot 
\left[\sqrt{2\pi}\left(\frac{5}{4} \right)^2\mathop{\sum_{l_{Q_1}}} \sqrt{l_{Q_1}}\right] \cdot 
\left[\sqrt{2\pi}\left(\frac{5}{4} \right)^2\mathop{\sum_{l_{R_1}}} \sqrt{l_{R_1}}\right] \times \\
\cdots \times
\left[\sqrt{2\pi}\left(\frac{5}{4} \right)^2\mathop{\sum_{l_{4}}} \sqrt{l_{4}}\right] \cdot
\left[\sqrt{2\pi}\left(\frac{5}{4} \right)^2 \mathop{\sum_{l_{3}}} \sqrt{l_{3}}\right] \cdot 
\left[\sqrt{2\pi}\left(\frac{5}{4} \right)^2\mathop{\sum_{l_{2}}} \sqrt{l_{2}}\right]\,.
\end{multline}
For each $Q_i$ and $R_i$, $i = 1,\ldots,\lceil \log_2 s \rceil - 2$, which means we have $\lceil \log_2 s \rceil - 2$ pairs plus one $Q_0$, hence \eqref{eqn:prob_1set_ssparseN} simplifies to the following.
\begin{align}
\label{eqn:prob_1set_ssparseO}
&\left[\sqrt{2\pi}\left(\frac{5}{4} \right)^2 \right]^{2\lceil \log_2 s \rceil - 3} \cdot
\left[\mathop{\sum_{l_{Q_0}}} \sqrt{l_{Q_0}}\right]\cdot 
\left[\mathop{\sum_{l_{Q_1}}} \sqrt{l_{Q_1}}\right]\cdot 
\cdots\left[\mathop{\sum_{l_{3}}} \sqrt{l_{3}} \right]\cdot 
\left[\mathop{\sum_{l_{2}}} \sqrt{l_{2}}\right] \\
\label{eqn:prob_1set_ssparseP}
\leq &\left[\sqrt{2\pi}\left(\frac{5}{4} \right)^2 \right]^{2\lceil \log_2 s \rceil - 3} \cdot
\left(q_0\sqrt{a_{Q_0}}\right)\cdot 
\left(q_1\sqrt{a_{Q_1}}\right)\cdots 
\left(q_{\lceil \log_2 s \rceil - 2}\sqrt{a_{3}} \right)\cdot 
\left(r_{\lceil \log_2 s \rceil - 2} \sqrt{a_{2}}\right) \\
\label{eqn:prob_1set_ssparseQ}
= &\left[\sqrt{2\pi}\left(\frac{5}{4} \right)^2 \right]^{2\lceil \log_2 s \rceil - 3} \cdot
\left(q_0q_1r_1 \cdots  q_{\lceil \log_2 s \rceil - 2}r_{\lceil \log_2 s \rceil - 2}\right) \cdot
\left(a_{Q_0}a_{Q_1}a_{R_1} \cdots a_{3}a_{2}\right)^{1/2}\,.
\end{align}
From \eqref{eqn:prob_1set_ssparseO} to \eqref{eqn:prob_1set_ssparseP} we upper each sum by taking the largest possible value of $l_{(\cdot)}$, which is $a_{(\cdot)}$, and multiplied it with the total number terms in the summation given by Lemma \ref{lem:num_size_split} in Appendix \ref{sec:oldresults}. We did upper bound the following two terms of \eqref{eqn:prob_1set_ssparseQ}.
\begin{align}
\label{eqn:prob_1set_ssparseQ1}
q_0q_1r_1q_2r_2q_3r_3 \cdots  q_{\lceil \log_2 s \rceil - 2}r_{\lceil \log_2 s \rceil - 2} & \leq s^{\log_2 s - 1} \,,\\
\label{eqn:prob_1set_ssparseQ2}
\left(a_{Q_0}a_{Q_1}a_{R_1}a_{Q_2}a_{R_2}a_{Q_3}a_{R_3} \cdots a_{3}a_{2}\right)^{1/2} & \leq 2^{-1} \cdot e^{\frac{1}{4}} \cdot \left(a_{s}\right)^{\log_2 s - \frac{3}{2}} \cdot s^{\frac{1}{2}\log_2 s + \frac{3}{2}}\,.
\end{align}
Details of the derivation of the bounds \eqref{eqn:prob_1set_ssparseQ1} and \eqref{eqn:prob_1set_ssparseQ2} is in the Appendix \ref{sec:inequalities}. Using these bounds from \eqref{eqn:prob_1set_ssparseQ} we have the following upper bound for $\Pi\left(l_s,\ldots,l_2,d\right)$.
\begin{align}
\label{eqn:prob_1set_ssparseR}
\Pi\left(l_s,\ldots,l_2,d\right) & \leq \left[\sqrt{2\pi}\left(\frac{5}{4} \right)^2 \right]^{2\lceil \log_2 s \rceil - 3} \cdot \left(s^{\log_2 s - 1} \right) \cdot \left(2^{-1} \cdot e^{\frac{1}{4}} \cdot \left(a_{s}\right)^{\log_2 s - \frac{3}{2}} \cdot s^{\frac{1}{2}\log_2 s + \frac{3}{2}} \right)\\
\label{eqn:prob_1set_ssparseS}
& \leq 2^{-3} \cdot e^{\frac{1}{4}} \cdot \left(a_{s}\right)^{\log_2 s - \frac{3}{2}} \cdot s^{\frac{3}{2}\log_2 s + \frac{9}{2}}\,.
\end{align}
From \eqref{eqn:prob_1set_ssparseR} to \eqref{eqn:prob_1set_ssparseS} we used the following upper bound.
\begin{equation}
\label{eqn:numbound}
\left[\sqrt{2\pi}\left(\frac{5}{4} \right)^2 \right]^{2\lceil \log_2 s \rceil - 3} \leq \left[\sqrt{2\pi}\left(\frac{5}{4} \right)^2 \right]^{2 \left(\log_2 s + 1 \right) - 3} 
\leq 4^{2 \log_2 s  - 1} = 2^{4 \log_2 s  - 2}  = 2^{-2} s^4\,.
\end{equation}
The bound \eqref{eqn:prob_1set_ssparseS} coincides with \eqref{eqn:pi_bound}, hence concluding the proof.
\vspace{-0.7cm}
\begin{flushright}
$\square$
\end{flushright}

\subsection{Theorem \ref{thm:prob_expansion_bound_new}} \label{sec:pthm3}
The following lemma is a key input in this proof.
\begin{lemma}
 \label{lem:beta}
 Let $0<\alpha\leq 1$, and $\varepsilon_n > 0$ such that $\varepsilon_n \rightarrow 0$ as $n\rightarrow \infty$. Then for $a_{s} < \hat{a}_{s}$, 
 \begin{equation}
 \label{eqn:set_sizes}
  a_{2i} = 2a_i + ca_i^2, \quad \mbox{for} \quad c = -\beta n^{-1}\,,
 \end{equation}
  where 
 \begin{equation}
  \label{eqn:beta}
  \beta = \frac{1 + \sqrt{1 - 4(1-\varepsilon_n)^2\left(1 - e^{-\alpha d} \right)e^{-\alpha d}}}{2(1-\varepsilon_n)\left(1 - e^{-\alpha d} \right)}\,.
 \end{equation}
\end{lemma}
The proof of the lemma is found in Section \ref{sec:plem4}. Recall from Theorem \ref{lem:prob_expansion_bound_new} that
\begin{equation}
\label{eqn:bigpsin_new3}
 \Psi_n\left(a_s,\ldots,a_1\right) = \frac{1}{n} \left[\frac{3\log 2}{2}\log_2^2 s + \left(\log_2 s - \frac{3}{2}\right)\log a_s + \sum_{i\in\Omega} \frac{s}{2i} \psi_i \right], ~~ \mbox{for} ~~ \Omega = \{2^j\}_{j=0}^{\log_2(s) - 1}\,,
\end{equation}
where 
\begin{equation}
\label{eqn:smallpsi2}
\psi_i = \left(n-a_i\right) \cdot \mathcal{H}\left(\frac{a_{2i}-a_i}{n-a_i}\right) + a_i\cdot \mathcal{H}\left(\frac{a_{2i}-a_i}{a_i}\right) - n\cdot \mathcal{H}\left(\frac{a_i}{n}\right)\,,
\end{equation} 

We use Lemma \ref{lem:beta} to upper bound $\psi_i$ in \eqref{eqn:smallpsi2} away from zero from above as $n\rightarrow 0$. We formalize this bound in the following proposition.
\begin{proposition}
 \label{pro:psibound}
 Let $\eta > 0$ and $\beta>1$ as defined in Lemma \ref{lem:beta}. Then 
 \begin{equation}
  \label{eqn:psibound}
  \psi_i \leq  -a_i\eta(\beta-1)\beta^{-1}\left(1-\beta\frac{a_i}{n}\right)^{-1}\,.
 \end{equation}
\end{proposition}
The proof of Proposition \ref{pro:psibound} is found in Section \ref{sec:ppro2}.
Using the bound of $\psi_i$ in Proposition \ref{pro:psibound}, we upper $\Psi\left(a_s,\ldots,a_1\right)$ as follows.
\begin{align}
\label{eqn:prob_expansion_bound4_new1}
 \Psi\left(a_s,\ldots,a_1\right) & \leq \frac{1}{n} \left[\frac{3\log 2}{2}\log_2^2 s + \left(\log_2 s - \frac{3}{2}\right)\log a_s - \sum_{i\in\Omega} \frac{s}{2i} \cdot \frac{a_i\eta(\beta-1)}{\beta\left(1-\beta\frac{a_i}{n}\right)} \right],  \\
 \label{eqn:prob_expansion_bound4_new2}
 & \leq - \frac{\eta(\beta-1)}{\beta}\sum_{i\in\Omega}\left( \frac{s}{2i} \cdot \frac{a_i}{n} \right)  + \frac{1}{n} \left[\frac{3\log 2}{2}\log_2^2 s + \left(\log_2 s - \frac{3}{2}\right)\log a_s \right]\,.
\end{align}
Then setting $a_s = (1-\e)ds$ and substituting in \eqref{eqn:prob_expansion_bound4_new2}, the factor multiplying $\frac{1}{n}$ becomes
\begin{align}
\label{eqn:prob_expansion_bound4_term1}
 & \frac{3\log 2}{2}\log_2^2 s + \left(\log_2 s - \frac{3}{2}\right)\log \left[(1-\e)ds \right]\\
 \label{eqn:prob_expansion_bound4_term2}
 & = \frac{3\log 2}{2}\log_2^2 s + \log (1-\e)\log_2 s + \log d \log_2 s + \log_2 s \log s - \frac{3}{2} \log\left[(1-\e)d\right] - \frac{3}{2} \log s\\
 \label{eqn:prob_expansion_bound4_term3}
 & = \frac{5\log 2}{2}\log_2^2 s + \left(\log_2 (1-\e) - 3/2\right)\log s + \log d \log_2 s + \log\left[(1-\e)^{-3/2}d^{-3/2}\right] \\
 \label{eqn:prob_expansion_bound4_term4}
 & = \frac{5\log 2}{2}\log_2^2 s + \log d \log_2 s + \log\left[(1-\e)^{-3/2}d^{-3/2}\right] + \log s^{\log_2(1-\e) - 3/2}\,.
\end{align}
The last two terms of \eqref{eqn:prob_expansion_bound4_term4} become polynomial in $s,d$ and $\e$, when exponentiated hence they are incorporated into $p_n(s,d,\e)$ in \eqref{eqn:pn_new2}, which means
\begin{align}
\label{eqn:pn_new3}
p_n(s,d,\e) & = p_n(s,d)\cdot \exp\left[\log\left[(1-\e)^{-3/2}d^{-3/2}\right] + \log s^{\log_2(1-\e) - 3/2}\right]\\
\label{eqn:pn_new4}
 & = 2^{-3}s^{9/2}e^{1/4} \cdot (1-\e)^{-3/2}d^{-3/2} s^{\log_2(1-\e) - 3/2}\\
  & = \frac{\sqrt[4]{e} \cdot s^{\log_2(1-\e) + 3}}{\sqrt{2^{6}(1-\e)^3d^3}} \,,
\end{align}
which is \eqref{eqn:pn_new2}.
The first two terms of \eqref{eqn:prob_expansion_bound4_term4} will grow faster than a polynomial in $s,d$ and $\e$ when exponentiated, hence they replace in \eqref{eqn:prob_expansion_bound4_new2}, the factor multiplying $\frac{1}{n}$. Therefore, 
\eqref{eqn:prob_expansion_bound4_term4} is modified as thus
\begin{equation}
\label{eqn:bigpsin_new3}
- \frac{\eta(\beta-1)}{\beta}\sum_{i\in\Omega}\left( \frac{s}{2i} \cdot \frac{a_i}{n} \right)  + \frac{1}{n}\left[\frac{5\log 2}{2}\log_2^2 s + \log d \log_2 s \right] =: \Psi_n\left(s,d,\e \right)\,.
\end{equation}
 
The factor $\sum_{i\in\Omega} \left(\frac{s}{2i} \cdot \frac{a_i}{n}\right)$ 
in \eqref{eqn:bigpsin_new3} is lower bounded as follows, see proof in Section \ref{sec:ineq1}.
\begin{equation}
 \label{eqn:seriesbound}
 \sum_{i\in\Omega}\left( \frac{s}{2i} \cdot \frac{a_i}{n}\right) \geq \frac{\log_2(s/2)}{2n}(1-\e)ds\,.
\end{equation}
Using this bound in \eqref{eqn:bigpsin_new3} gives \eqref{eqn:bigpsin_new2},
thus concluding the proof.
\vspace{-0.7cm}
\begin{flushright}
$\square$
\end{flushright}

\subsection{Lemma \ref{lem:beta}} \label{sec:plem4}
 Recall that we have a formula for the expected values of the $a_i$ as
\begin{equation}
\label{eqn:expectedvalues}
 \hat{a}_{2i} = \hat{a}_{i} \left(2 - \frac{\hat{a}_{i}}{n} \right) \quad \mbox{for} \quad i\in\{2^j\}_{j=0}^{\log_2(s) - 1},
\end{equation}
which follow a relatively simple formulas, and then the coupled system of cubics as
\begin{equation}
\label{eqn:coupledsystem}
 a_{2i}^3 - 2a_{i}a_{2i}^2 + 2a_{i}^2a_{2i} - a_{i}^2a_{4i} = 0 \quad \mbox{for} \quad i\in\{2^j\}_{j=0}^{\log_2(s) - 2},
\end{equation}
for when the final $a_s$ is constrained to be less than $\hat{a}_s$. To simplify the notation of the indexing in \eqref{eqn:coupledsystem}, observe that if $i = 2^{j}$ for a fixed $j$, then $2i = 2^{j+1}$ and $4i = 2^{j+2}$. Therefore, it suffice to use the index $a_j, a_{j+1}$, and $a_{j+2}$ rather than $a_{i}, a_{2i}$, and $a_{4i}$. Moving the second two terms in \eqref{eqn:coupledsystem} to the right and dividing the quadratic multiples we get the relation
\begin{equation}
\label{eqn:coupledsystemordered}
 \frac{a_{j+2}-2a_{j+1}}{a_{j+1}^2}=\frac{a_{j+1}-2a_j}{a_j^2},
\end{equation}
which is the same expression on the right and left, but with $j$ increased by one on the left.  This implies that the fraction is independent of $j$, so
\begin{equation}
\label{eqn:constantrelation}
 \frac{a_{j+1}-2a_j}{a_j^2}=c\,, \quad \Rightarrow \quad a_{j+1}= a_j(2 + c a_j)\,,
\end{equation}
for some constant $c$ independent of $j$ (though not necessarily of $n$). This is in fact the relation \eqref{eqn:expectedvalues}, if we set $c$ to be equal to $-1/n$. One can then wonder what is the behavior of $c$ if we fix the final $a_s$. Moreover, \eqref{eqn:constantrelation} is equivalent to
\begin{equation}
\label{eqn:constantrelation2}
c a_{j+1} + 1 = (c a_j + 1)^2\,,
\end{equation}
which inductively leads to
\begin{equation}
\label{eqn:constantrelation3}
c a_{l} + 1 = (c a_0 + 1)^{2^l}, \quad l>0,
\end{equation}
so that one has a relation of the $l^{th}$ stage in terms of the first stage. 
Note this does not require the $a_s$ to be fixed, \eqref{eqn:constantrelation3} is how one simply computes all $a_l$ for $l>0$ once one has $a_0$ and $c$.  The point is that $c$ to match the $a_s$ one has to select $c$ appropriately.  So the way we calculate $c$ is by knowing $a_0$ and $a_s$, then solving \eqref{eqn:constantrelation3} for $l=s$.  
Unfortunately there is not an easy way to solve for $c$ in \eqref{eqn:constantrelation3} so we need to do some asymptotic approximation. 
Let's assume that $a_l$ is close to $\hat{a}_l$.  So we do an asymptotic expansion in terms of the difference from $\hat{a}_l$. 

To simplify things a bit lets insert $a_0=d$ (since $a_0$ is $a_1$ in our standard notation) and then we insert what we know for $\hat{a}_l$.  For $\hat{a}_l$ we have $c=-n^{-1}$, see \eqref{eqn:expectedvalues}.  We then have from \eqref{eqn:constantrelation3} that 
\begin{equation}
\label{eqn:constantrelation4}
a_l = c^{-1}  (cd+1)^{2^l} - c^{-1} \quad \mbox{and} \quad \hat{a}_l = -n (-d/n+1)^{2^l} + n.
\end{equation}
So if we write $a_l = (1-\varepsilon_n)\hat{a}_l$ and consider the case of $\varepsilon_n \rightarrow 0$ as $n\rightarrow 0$. The point of this is that instead of working with $\hat{a}_l$ we can now work in terms of $\varepsilon_n$.
Setting $a_l = (1-\varepsilon_n)\hat{a}_l$ gives
\begin{equation}
\label{eqn:constantrelation5}
c^{-1} (cd+1)^{2^l} - c^{-1} = -n(1-\varepsilon_n) \left[(-d/n+1)^{2^l} - 1\right]\,.
\end{equation}
We now solve for $c$ as a function of $\varepsilon_n$ and $d$.  As $\varepsilon_n$ goes to zero we should have $c$ converging to $-n^{-1}$.  

Let $\alpha n = 2^l$, for $0 < \alpha \leq 1$ and $c = -\beta(\varepsilon_n,d)/n$, then, dropping the argument of $\beta(\cdot,\cdot)$, \eqref{eqn:constantrelation5} becomes
\begin{equation}
\label{eqn:constantrelation6}
\frac{n}{\beta} \left(1-\frac{\beta d}{n}\right)^{\alpha n} - \frac{n}{\beta} = n (1-\varepsilon_n)\left(1-\frac{d}{n}\right)^{\alpha n} - n(1-\epsilon)\,.
\end{equation}
Multiplying through by $\beta/n$ and performing a change of variables of $k = \alpha n$, \eqref{eqn:constantrelation6} becomes
\begin{equation}
\label{eqn:constantrelation7}
\left(1-\frac{\alpha \beta d}{k}\right)^{k} - 1 = \beta(1-\varepsilon_n)\left(1-\frac{\alpha d}{k}\right)^{k} - \beta(1-\varepsilon_n)\,.
\end{equation}
The left hand side of \eqref{eqn:constantrelation7} simplifies to 
\begin{equation}
\label{eqn:constantrelation7lhs}
e^{-\alpha \beta d} - 1 -\frac{e^{-\alpha \beta d}\alpha^2 \beta^2d^2}{2k} + \bigO(k^{-2})\,.
\end{equation}
The right hand side of \eqref{eqn:constantrelation7} simplifies to 
\begin{equation}
\label{eqn:constantrelation7rhs}
\beta(1-\varepsilon_n)e^{-\alpha d} - \beta(1-\varepsilon_n) - \frac{\beta(1-\varepsilon_n)e^{-\alpha d}\alpha^2d^2}{2k} + \bigO(k^{-2})\,.
\end{equation}
Matching powers of $k$ in \eqref{eqn:constantrelation7lhs} and \eqref{eqn:constantrelation7rhs} for $k^0$ and $k^{-1}$ yields the following.
\begin{align}
\label{eqn:constantrelation7n0}
e^{-\alpha \beta d} - 1	& = \beta(1-\varepsilon_n)e^{-\alpha d} - \beta(1-\varepsilon_n) \,,\quad \mbox{and} \\ 
\label{eqn:constantrelation7n1}
\alpha^2 \beta^2d^2e^{-\alpha \beta d} & = \beta(1-\varepsilon_n)\alpha^2d^2e^{-\alpha d}\,.
\end{align}
Both of which respectively simplify to the following.
\begin{align}
\label{eqn:constantrelation7n0b}
e^{-\alpha \beta d} - 1 & = \beta(1-\varepsilon_n)\left(e^{-\alpha d}  - 1 \right) \,,\quad \mbox{and} \\ 
\label{eqn:constantrelation7n1b}
\beta e^{-\alpha \beta d} & = (1-\varepsilon_n)e^{-\alpha d}\,.
\end{align}
Multiply \eqref{eqn:constantrelation7n0b} by $\beta$ and subtract the two equations, \eqref{eqn:constantrelation7n0b} and \eqref{eqn:constantrelation7n1b}, to get
\begin{equation}
\label{eqn:constantrelation7n1a}
(1-\varepsilon_n)\left(1 - e^{-\alpha d} \right) \beta^2 - \beta + (1-\varepsilon_n)e^{-\alpha d} = 0\,.
\end{equation}
This yields
\begin{equation}
\label{eqn:constantrelation8}
\beta = \frac{1 \pm \sqrt{1 - 4(1-\varepsilon_n)^2\left(1 - e^{-\alpha d} \right)e^{-\alpha d}}}{2(1-\varepsilon_n)\left(1 - e^{-\alpha d} \right)}\,.
\end{equation}
To be consistent with what $c$ ought to be as $\varepsilon_n \rightarrow 0$, we choose 
\begin{equation}
\label{eqn:fepsilon}
\beta(\varepsilon_n,d) = \frac{1 + \sqrt{1 - 4(1-\varepsilon_n)^2\left(1 - e^{-\alpha d} \right)e^{-\alpha d}}}{2(1-\varepsilon_n)\left(1 - e^{-\alpha d} \right)}\,,
\end{equation}
as required -- concluding the proof.
\vspace{-0.7cm}
\begin{flushright}
$\square$
\end{flushright}

\subsection{Proposition \ref{pro:psibound}} \label{sec:ppro2}
We use Lemma \ref{lem:beta} to express $\psi_i$ in \eqref{eqn:smallpsi2} as follows
\begin{align}
\label{eqn:prob_expansion_bound3_new1}
 \psi_i & = -n\left[\mathcal{H}\left(\frac{a_i}{n}\right) - \mathcal{H}\left(\frac{a_{i}+ca^2_i}{n-a_i}\right)\right] + a_i\left[ \mathcal{H}\left(\frac{a_{i}+ca^2_i}{a_i}\right) - \mathcal{H}\left(\frac{a_{i}+ca^2_i}{n-a_i}\right)\right] \\
 \label{eqn:prob_expansion_bound3_new2}
  & = -n\left[\mathcal{H}\left(\frac{a_i}{n}\right) - \mathcal{H}\left(\frac{a_i}{n}\cdot\frac{1+ca_i}{1-\frac{a_i}{n}}\right)\right] + a_i\left[ \mathcal{H}\left(1+ca_i\right) - \mathcal{H}\left(\frac{a_i}{n}\cdot\frac{1+ca_i}{1-\frac{a_i}{n}}\right)\right]\,.
\end{align}
Note  that for regimes of small $s/n$ considered
\begin{equation}
 \label{eqn:ordering}
 -\frac{\beta}{n} = c \leq -\frac{1}{n} , \quad \Rightarrow \quad ca_i \leq -\frac{a_i}{n}, \quad \mbox{and} \quad 1+ca_i \leq 1 - \frac{a_i}{n} \,.
\end{equation}
We need the following expressions for the Shannon entropy and it's first and second derivatives
\begin{align}
 \label{eqn:sentrpy}
 \mathcal{H}(z) & = -z\log z - (1-z)\log(1-z), \\
 \label{eqn:sentrpyp}
 \mathcal{H}'(z) & = \log \left(\frac{1-z}{z}\right), \quad \mbox{and} \\
 \label{eqn:sentrpypp}
 \mathcal{H}''(z) & = -\frac{1}{z(1-z)}\,.
\end{align}
But also $\mathcal{H}(z) = \mathcal{H}(1-z)$ due to the symmetry about $z = 1/2$. Similarly, $\mathcal{H}''(z)$ is symmetric about $z=1/2$; while $\mathcal{H}'(z)$ is anti-symmetric, i.e. $\mathcal{H}'(z) = -\mathcal{H}'(1-z)$.
Using the symmetry of $\mathcal{H}(z)$ we rewrite $\psi_i$ in \eqref{eqn:prob_expansion_bound3_new2} as follows.
\begin{equation}
\label{eqn:prob_expansion_bound3_new3}
 \psi_i = -n\left[\mathcal{H}\left(\frac{a_i}{n}\right) - \mathcal{H}\left(\frac{a_i}{n}\cdot\frac{1+ca_i}{1-\frac{a_i}{n}}\right)\right] + a_i\left[ \mathcal{H}\left(-ca_i\right) - \mathcal{H}\left(\frac{a_i}{n}\cdot\frac{1+ca_i}{1-\frac{a_i}{n}}\right)\right]\,.
\end{equation}
From \eqref{eqn:ordering}, we deduce the following ordering
\begin{equation}
 \label{eqn:ordering2}
 \frac{a_i}{n}\cdot\frac{1+ca_i}{1-\frac{a_i}{n}} \leq \frac{a_i}{n} \leq -ca_i \leq 1/2 \,.
\end{equation}
To simplify notation, let $x_1 = \frac{a_i}{n}\cdot\frac{1+ca_i}{1-\frac{a_i}{n}}$, $x_2 = \frac{a_i}{n}$, and $x_3  = -ca_i$, which implies that $x_1 \leq x_2 \leq x_3 \leq 1/2$. Therefore, from  \eqref{eqn:prob_expansion_bound3_new3}, we have
\begin{align}
\label{eqn:prob_expansion_bound3_new4}
 \psi_i & = -n\left[\mathcal{H}\left(x_2\right) - \mathcal{H}\left(x_1\right)\right] + a_i\left[ \mathcal{H}\left(x_3\right) - \mathcal{H}\left(x_1\right)\right]\\
 \label{eqn:prob_expansion_bound3_new5}
  & = -n\left[\mathcal{H}\left(x_2\right) - \mathcal{H}\left(x_1\right)\right] + a_i\left[ \mathcal{H}\left(x_3\right) - \mathcal{H}\left(x_2\right) + \mathcal{H}\left(x_2\right) - \mathcal{H}\left(x_1\right)\right]\\
  \label{eqn:prob_expansion_bound3_new6}
  & = -(n-a_i)\left[\mathcal{H}\left(x_2\right) - \mathcal{H}\left(x_1\right)\right] + a_i\left[ \mathcal{H}\left(x_3\right) - \mathcal{H}\left(x_2\right)\right]\\
  \label{eqn:prob_expansion_bound3_new7}
  & \leq -(n-a_i)(x_2-x_1)\mathcal{H}'\left(x_{2}\right) + a_i(x_3-x_2)\mathcal{H}'\left(x_{2}\right)\\
  \label{eqn:prob_expansion_bound3_new7b}
  & = \left[a_i(x_3-x_2) - (n-a_i)(x_2-x_1)\right]\mathcal{H}'\left(x_{2}\right)\,.
\end{align}
Observe that the expression in the square brackets on the right hand side of \eqref{eqn:prob_expansion_bound3_new7b} is zero, which implies that
\begin{equation}
\label{eqn:nicerelation}
 a_i(x_3-x_2) = (n-a_i)(x_2-x_1)\,.
\end{equation}
This is very easy to check by substituting the values of $x_1$, $x_2$, and $x_3$.
So instead of bound \eqref{eqn:prob_expansion_bound3_new7}, we alternatively upper bound \eqref{eqn:prob_expansion_bound3_new6} as follows
\begin{equation}
\label{eqn:prob_expansion_bound3_new8}
 \psi_i \leq -(n-a_i)(x_2-x_1)\mathcal{H}'\left(\xi_{21}\right) + a_i(x_3-x_2)\mathcal{H}'\left(\xi_{32}\right)\,,
\end{equation}
where $\xi_{21} \in (x_1,x_2)$, and $\xi_{32} \in (x_2,x_3)$, which implies
\begin{equation}
\label{eqn:nicerelation2}
  \xi_{21} < \xi_{32}, \quad \mbox{and} \quad \mathcal{H}'\left(\xi_{21}\right) > \mathcal{H}'\left(\xi_{32}\right)\,.
\end{equation}
Using relation \eqref{eqn:nicerelation}, bound \eqref{eqn:prob_expansion_bound3_new8} simplifies to the following.
\begin{align}
\label{eqn:prob_expansion_bound3_new9}
 \psi_i & \leq -a_i(x_3-x_2)\mathcal{H}'\left(\xi_{21}\right) + a_i(x_3-x_2)\mathcal{H}'\left(\xi_{32}\right)\\
 \label{eqn:prob_expansion_bound3_new10}
  & = -a_i(x_3-x_2)\left[\mathcal{H}'\left(\xi_{21}\right) - \mathcal{H}'\left(\xi_{32}\right)\right]\\
  \label{eqn:prob_expansion_bound3_new11}
  & \leq -a_i(x_3-x_2)\left(\xi_{21}-\xi_{32}\right)\mathcal{H}''\left(\xi_{31}\right)\,,
\end{align}
for $\xi_{31} \in \left(x_{1},x_{3}\right)$. Since $\xi_{21} < \xi_{32}$, we rewrite bound \eqref{eqn:prob_expansion_bound3_new11} as follows.
\begin{align}
\label{eqn:prob_expansion_bound3_new12}
 \psi_i & \leq a_i(x_3-x_2)\left(\xi_{32} - \xi_{21}\right)\mathcal{H}''\left(\xi_{31}\right)\\
 \label{eqn:prob_expansion_bound3_new13}
  & \leq a_i\eta(x_3-x_2)\mathcal{H}''\left(x_{3}\right)\,,
\end{align}
where $\eta = \xi_{32} - \xi_{21} > 0$, and the last bound is due to the fact that $x_{3} > \xi_{31}$. 

Going back to our normal notation, we rewrite bound \eqref{eqn:prob_expansion_bound3_new13} as follows.
\begin{align}
\label{eqn:prob_expansion_bound3_new14}
 \psi_i & \leq a_i\eta\left(-ca_i-\frac{a_i}{n}\right)\mathcal{H}''\left(-ca_i\right)\\
 \label{eqn:prob_expansion_bound3_new15}
  & = a_i\eta\frac{a_i}{n}(\beta-1)\frac{-1}{\beta\frac{a_i}{n}\left(1-\beta\frac{a_i}{n}\right)}\\
 \label{eqn:prob_expansion_bound3_new16}
  & = -\frac{a_i\eta(\beta-1)}{\beta\left(1-\beta\frac{a_i}{n}\right)}\,,
\end{align}
This concludes the proof.
\vspace{-0.7cm}
\begin{flushright}
$\square$
\end{flushright}

\subsection{Inequality \ref{eqn:seriesbound}} \label{sec:ineq1}
The series bound \eqref{eqn:seriesbound} is derived as follows.
\begin{align}
 \label{eqn:seriesbound1}
 \sum_{i\in\Omega}\left( \frac{s}{2i} \cdot \frac{a_i}{n}\right) & = \left(\frac{s}{2} \cdot \frac{a_1}{n}\right) + \left(\frac{s}{4} \cdot \frac{a_2}{n}\right) + \cdots + \left(\frac{s}{s} \cdot \frac{a_{s/2}}{n}\right) \\
 \label{eqn:seriesbound2}
 & \geq \left(\frac{s}{2n} \cdot \frac{(1-\e)ds}{2^{\log_2 s}}\right) + \left(\frac{s}{4n} \cdot \frac{(1-\e)ds}{2^{\log_2(s/2)}} \right) + \cdots + \left(\frac{s}{sn} \cdot \frac{(1-\e)ds}{2}\right) \\
 \label{eqn:seriesbound3}
 & = \left[\left(\frac{s}{2n}\cdot\frac{1}{s}\right) + \left(\frac{s}{4n}\cdot\frac{2}{s}\right) + \cdots + \left(\frac{s}{sn}\cdot\frac{1}{2}\right) \right](1-\e)ds \\
 \label{eqn:seriesbound4}
 & = \left(\frac{1}{2n} + \frac{1}{2n} + \cdots + \frac{1}{2n}\right)(1-\e)ds ~ = ~ \frac{\log_2(s/2)}{2n}(1-\e)ds\,.
\end{align}
That is the required bounds, hence concluding the proof.
\vspace{-0.7cm}
\begin{flushright}
$\square$
\end{flushright}

%

\section{Conclusion} \label{sec:conclusion}
We considered the construction of sparse matrices that are invaluable for dimensionality reduction with application in diverse fields. These sparse matrices are more efficient computationally compared to their dense counterparts also used for the purpose of dimensionality reduction. Our construction is probabilistic based on the dyadic splitting method we introduced in \cite{bah2013vanishingly}. By better approximation of the bounds we achieve a novel result, which is a reduced complexity of the sparsity per column of these matrices. Precisely, a complexity that is a state-of-the-art divided by $\log s$, where $s$ is the intrinsic dimension of the problem. 

Our approach is one of a few that gives quantitative sampling theorems for existence of such sparse matrices. Moreover, using the phase transition framework comparison, our construction is better than existing probabilistic constructions. We are also able to compare performance of combinatorial compressed sensing algorithms by comparing their phase transition curves. This is one perspective in algorithm comparison amongst a couple of others like runtime and iteration complexities.

Evidently, our results holds true for the construction of expander graphs, which is a graph theory problem and is of interest to communities in theoretical computer science and pure mathematics.

\vspace{0.5cm}
\begin{center}
{\bf ACKNOWLEDGMENT}
\end{center} \vspace{-0.2cm}
BB acknowledges the support from the funding by the German Federal Ministry of Education and Research (BMBF) for the German Research Chair at AIMS South Africa, funding for which is administered by Alexander von Humboldt Foundation (AvH). JT acknowledges support from The Alan Turing Institute under the EPSRC grant EP/N510129/1.

\section{Appendix} \label{sec:appdx}
\subsection{Key relevant results from \cite{bah2013vanishingly}} \label{sec:oldresults}
In order to make this manuscript self containing we include in this section key relevant lemmas, corollaries and definitions from \cite{bah2013vanishingly}. 

\begin{lemma}[Lemma 2.5, \cite{bah2013vanishingly}]
\label{lem:num_size_split}
Let $S$ be an index set of cardinality $s$.  For any level $j$ of the dyadic splitting, $j = 0,\ldots,\lceil \log_2 s \rceil -1$, the set $S$ is decomposed into disjoint sets each having cardinality $Q_j = \big{\lceil}\frac{s}{2^j}\big{\rceil}$ or $R_j = Q_j - 1$. Let $q_j$ sets have cardinality $Q_j$ and $r_j$ sets have cardinality $R_j$, then
\begin{align}
\label{eqn:split_num_size}
q_j = s - 2^j\cdot\Big{\lceil}\frac{s}{2^j}\Big{\rceil} + 2^j, \quad \mbox{and} \quad r_j = 2^j - q_j.
\end{align}
\end{lemma}

\begin{lemma}[Lemma 2.3, \cite{bah2013vanishingly}]
\label{lem:intersect_prob}
Let $B,~B_1,~B_2 \subset [n]$ where $\left|B_1\right|=b_1$, $\left|B_2\right|=b_2$, $B = B_1\cup B_2$ and $|B|=b$. Also let $B_1$ and $B_2$ be drawn uniformly at random, independent of each other, and define $\hbox{P}_n\left(b,b_1,b_2\right) := \hbox{Prob}\left(\left|B_1\cap B_2\right| = b_1+b_2-b\right)$, then 
\begin{equation}
\label{eqn:intersect_prob}
\hbox{P}_n\left(b,b_1,b_2\right) = \binom{b_1}{b_1+b_2-b} \binom{n-b_1}{b-b_1} \binom{n}{b_2}^{-1}.
\end{equation}
\end{lemma}

\begin{definition}
\label{def:prob_1set_poly}
$\hbox{P}_n\left(x,y,z\right)$ defined in \eqref{eqn:intersect_prob} satisfies
the upper bound
\begin{equation}
\label{eqn:prob_1set_poly1}
\hbox{P}_n\left(x,y,z\right)\le \pi \left( x,y,z \right) \exp(\psi_n(x,y,z))
\end{equation}
with bounds of $\pi \left( x,y,z \right)$ given in Lemma \ref{lem:prob_1set_poly}.
\end{definition}

\begin{lemma}
\label{lem:prob_1set_poly}
For $\pi \left( x,y,z \right)$ and $\hbox{P}_n\left(x,y,z\right)$ given by \eqref{eqn:prob_1set_poly1} and \eqref{eqn:intersect_prob} respectively, if $\{y,z\}<x<y+z$, $\pi \left( x,y,z \right)$ is given by
\begin{equation}
\label{eqn:prob_1set_poly3a}
 \left(\frac{5}{4}\right)^4 \left[ \frac{yz(n-y)(n-z)}{2\pi n(y+z-x)(x-y)(x-z)(n-x)} \right]^{\frac{1}{2}},
\end{equation}
otherwise $\pi \left( x,y,z \right)$ has the following cases.
\begin{align}
\label{eqn:prob_1set_poly3b}
\left( \frac{5}{4} \right)^3 \left[ \frac{y(n-z)}{n(y-z)} \right]^{\frac{1}{2}} & \quad \mbox{if} \quad x=y>z;\\
\label{eqn:prob_1set_poly3c}
\left( \frac{5}{4} \right)^3 \left[ \frac{(n-y)(n-z)}{n(n-y-z)} \right]^{\frac{1}{2}} & \quad \mbox{if} \quad x=y+z;\\
\label{eqn:prob_1set_poly3d}
\left( \frac{5}{4} \right)^2 \left[ \frac{2\pi z(n-z)}{n} \right]^{\frac{1}{2}} & \quad \mbox{if} \quad x=y=z.
\end{align}
\end{lemma}

\begin{lemma}
\label{lem:psi_n_behaviour}
Define
\begin{equation}
\label{eqn:psi_n_def}
\psi_n(x,y,z) :=y\cdot\hbox{H}\left(\frac{x-z}{y}\right) + (n-y)\cdot\hbox{H}\left(\frac{x-y}{n-y}\right) - n\cdot\hbox{H}\left(\frac{z}{n}\right)\,,
\end{equation}
then for $n>x>y$ we have that
\begin{align}
\label{eqn:psi_n_property1}
& \mbox{for } y>z \quad \psi_n(x,y,y) \leq \psi_n(x,y,z) \leq \psi_n(x,z,z); \\
\label{eqn:psi_n_property2}
& \mbox{for } x>z \quad \psi_n(x,y,y) > \psi_n(z,y,y); \\
 \label{eqn:psi_n_property3}
 & \mbox{for } 1/2< \a \leq 1 \quad \psi_n(x,y,y) < \psi_n(\a x,\a y,\a y).
\end{align}
\end{lemma}

\begin{corollary}
\label{cor:pi_monotonicity}
If $n>2y$, then $\pi(y,y,y)$ is monotonically increasing in $y$.
\end{corollary}

The following bound, used in \cite{bah2013vanishingly}, is deducible from an asymptotic series for the logarithms Stirling approximation of the factorial ($!$)
\begin{equation}
 \label{eqn:stirling}
 \frac{16 e^{N\mathcal{H}(p)}}{25\sqrt{2\pi p(1-p)N}} \leq \binom{N}{pN} \leq \frac{5 e^{N\mathcal{H}(p)}}{4\sqrt{2\pi p(1-p)N}} \,.
\end{equation}

\subsection{Derivation of Inequalities} \label{sec:inequalities}
\label{sec:bounds}
\subsubsection{Inequality \ref{eqn:prob_1set_ssparseQ1}} \label{sec:bounds1}
By Lemma \ref{lem:num_size_split}, the left hand side (LHS) of \eqref{eqn:prob_1set_ssparseQ1} is equal to the following.
\begin{multline}
\label{eqn:prob_1set_ssparseQ1a}
q_0\left(q_1r_1\right) \cdot \left(q_2r_2\right) \cdot \left(q_3r_3\right)  \cdots  \left(q_{\lceil \log_2 s \rceil - 2}r_{\lceil \log_2 s \rceil - 2}\right) = \left( s - \Big{\lceil}\frac{s}{1}\Big{\rceil} + 1\right) \cdot \left(s - 2\cdot\Big{\lceil}\frac{s}{2}\Big{\rceil} + 2 \right) \\
\times \left(2 - \left(s - 2\cdot\Big{\lceil}\frac{s}{2}\Big{\rceil} + 2 \right) \right) \times \cdots \times \left(s - 2^{\lceil \log_2 s\rceil - 2}\cdot\Big{\lceil}\frac{s}{2^{\lceil \log_2 s\rceil - 2}}\Big{\rceil} + 2^{\lceil \log_2 s\rceil - 2} \right) \\ 
\times \left(2^{\lceil \log_2 s\rceil - 2} - \left(s - 2^{\lceil \log_2 s\rceil - 2}\cdot \Big{\lceil}\frac{s}{2^{\lceil \log_2 s\rceil - 2}}\Big{\rceil} + 2^{\lceil \log_2 s\rceil - 2} \right) \right)\,.
\end{multline}

We simplify \eqref{eqn:prob_1set_ssparseQ1a} to get the following.
\begin{multline}
\label{eqn:prob_1set_ssparseQ1b}
 1 \cdot \left(s - 2\cdot\Big{\lceil}\frac{s}{2}\Big{\rceil} + 2 \right) \cdot \left(2\cdot\Big{\lceil}\frac{s}{2}\Big{\rceil} - s \right) \cdot \left(s - 2^2\cdot\Big{\lceil}\frac{s}{2^2}\Big{\rceil} + 2^2 \right) \cdot \left(2^2\cdot\Big{\lceil}\frac{s}{2^2}\Big{\rceil} - s \right) \times \\
 \cdots \times \left(s - 2^{\lceil \log_2 s\rceil - 2}\cdot\Big{\lceil}\frac{s}{2^{\lceil \log_2 s\rceil - 2}}\Big{\rceil} + 2^{\lceil \log_2 s\rceil - 2} \right) \cdot \left(2^{\lceil \log_2 s\rceil - 2}\cdot\Big{\lceil}\frac{s}{2^{\lceil \log_2 s\rceil - 2}}\Big{\rceil} - s \right)\,.
\end{multline}

We upper bound $- \lceil z \rceil$ by $-z$ and $\lceil z \rceil$ by $z+1$ to upper bound \eqref{eqn:prob_1set_ssparseQ1b} as follows.
\begin{multline}
\label{eqn:prob_1set_ssparseQ1c}
\left(s - 2\cdot\frac{s}{2} + 2 \right) \cdot \left(2\left(\frac{s}{2} + 1\right) - s \right) \cdot \left(s - 4\cdot\frac{s}{4} + 2 \right) \cdot \left(4\left(\frac{s}{4} + 1\right) - s \right) \times \\
 \cdots \times \left(s - 2^{\lceil \log_2 s\rceil - 2} \cdot \frac{s}{2^{\lceil \log_2 s\rceil - 2}} + 2^{\lceil \log_2 s\rceil - 2} \right) \cdot \left(2^{\lceil \log_2 s\rceil - 2}\cdot \frac{s}{2^{\lceil \log_2 s\rceil - 2}} - s \right)\,.
\end{multline}

The bound \eqref{eqn:prob_1set_ssparseQ1c} is then simplified to the following.
\begin{align}
\label{eqn:prob_1set_ssparseQ1d}
(2\cdot 2) \cdot (4\cdot 4) \cdot (8\cdot 8) \times \cdots \times \left(2^{\lceil \log_2 s\rceil - 2} \cdot 2^{\lceil \log_2 s\rceil - 2} \right) & = 2^2 \cdot 4^2 \cdot 8^2 \times \cdots \times 2^{2\lceil \log_2 s\rceil - 4}\\
\label{eqn:prob_1set_ssparseQ1e}
& = 4^1 \cdot 4^2 \cdot 4^3 \cdots \times 4^{\lceil \log_2 s\rceil - 2} \\
\label{eqn:prob_1set_ssparseQ1f}
& \leq 4^{\left(\sum_{i=1}^{\log_2 s - 1} i\right)} \\
\label{eqn:prob_1set_ssparseQ1g}
& = 4^{\frac{1}{2}\left(\log_2 s - 1\right)\cdot \log_2 s} = 2^{\left(\log_2 s - 1\right)\cdot \log_2 s} \,.
\end{align}
In \eqref{eqn:prob_1set_ssparseQ1f} we upper bound $\lceil \log_2 s\rceil$ by $\log_2 s + 1$; while in the LHS of \eqref{eqn:prob_1set_ssparseQ1g} we computed the summation of a finite arithmetic series. After some algebraic manipulations of logarithms we end up with the RHS of \eqref{eqn:prob_1set_ssparseQ1g}, which simplifies to \eqref{eqn:prob_1set_ssparseQ1}.

\subsubsection{Inequality \ref{eqn:prob_1set_ssparseQ2}} \label{sec:bounds2}
Again by Lemma \ref{lem:num_size_split}, the left hand side (LHS) of \eqref{eqn:prob_1set_ssparseQ2}, i.e. $\left(a_{Q_0}a_{Q_1}a_{R_1}a_{Q_2}a_{R_2}a_{Q_3}a_{R_3} \cdots a_{3}a_{2}\right)^{1/2} $ is equal to the following.
\begin{equation}
\label{eqn:prob_1set_ssparseQ2a}
\left(a_{\lceil \frac{s}{2^0}\rceil}a_{\lceil \frac{s}{2^1}\rceil}a_{\lceil \frac{s}{2^1}\rceil - 1}a_{\lceil \frac{s}{2^2}\rceil}a_{\lceil \frac{s}{2^2}\rceil - 1}a_{\lceil \frac{s}{2^3}\rceil}a_{\lceil \frac{s}{2^3}\rceil - 1} \times \cdots \times a_{\lceil\frac{s}{2^{\lceil \log_2 s\rceil - 2}}\rceil}a_{\lceil\frac{s}{2^{\lceil \log_2 s\rceil - 2}}\rceil - 1}\right)^{1/2}\,.
\end{equation}

Given the monotonicity of $a_{(\cdot)}$ in terms of its subscripts, which indicate cardinalities of sets. Due to the nestedness of the sets due to the dyadic splitting, we upper bound  $a_{\lceil \frac{s}{2^j}\rceil - 1}$ by $a_{\frac{s}{2^j}}$, and $a_{\lceil \frac{s}{2^j}\rceil}$ by $a_{\frac{s}{2^j} + 1}$, resulting in the following upper bound for \eqref{eqn:prob_1set_ssparseQ2a}.
\begin{align}
\label{eqn:prob_1set_ssparseQ2b}
& \left[a_{s}a_{\left(\frac{s}{2} + 1\right)}a_{\frac{s}{2}}a_{\left(\frac{s}{4} + 1\right)}a_{\frac{s}{4}}a_{\left(\frac{s}{8} + 1\right)}a_{\frac{s}{8}}\times \cdots \times a_{\left(\frac{s}{2^{\lceil \log_2 s\rceil - 2}} + 1\right)}a_{\frac{s}{2^{\lceil \log_2 s\rceil - 2}}}\right]^{1/2}\\
\label{eqn:prob_1set_ssparseQ2c}
& \leq \left[a_{s}a_{\left(\frac{s}{2} + 1\right)}a_{\frac{s}{2}}a_{\left(\frac{s}{4} + 1\right)}a_{\frac{s}{4}}a_{\left(\frac{s}{8} + 1\right)}a_{\frac{s}{8}}\times \cdots \times a_{\left(\frac{s}{2^{\log_2 s - 2}} + 1\right)}a_{\frac{s}{2^{\log_2 s - 2}}}\right]^{1/2}\,.
\end{align}

In \eqref{eqn:prob_1set_ssparseQ2c} we used the fact that $2^{\log_2 s - 2}$ is a lower bound to $2^{\lceil \log_2 s\rceil - 2}$. We fix $a_s = (1-\e)ds =: cs$ and we require expansion to hold for all $|\support| \leq s$, i.e. $a_{s'} = cs'$ for all $s'\leq s$. Thus we can re-write \eqref{eqn:prob_1set_ssparseQ2c} as follows.
\begin{align}
\label{eqn:prob_1set_ssparseQ2d}
& \left[a_{s} \left(\frac{cs}{2} + c\right)\frac{cs}{2} \left(\frac{cs}{4} + c\right) \frac{cs}{4} \left(\frac{cs}{8} + c\right) \frac{cs}{8}\times \cdots \times \left(\frac{cs}{2^{\log_2 s - 2}} + c\right) \frac{cs}{2^{\log_2 s - 2}}\right]^{1/2}\\
\label{eqn:prob_1set_ssparseQ2e}
& = \left[a_{s} \left(\frac{a_{s}}{2} + c\right)\frac{a_{s}}{2} \left(\frac{a_{s}}{4} + c\right) \frac{a_{s}}{4} \left(\frac{a_{s}}{8} + c\right) \frac{a_{s}}{8}\times \cdots \times \left(\frac{a_{s}}{2^{\log_2 s - 2}} + c\right) \frac{a_{s}}{2^{\log_2 s - 2}}\right]^{1/2}\,.
\end{align}

In \eqref{eqn:prob_1set_ssparseQ2e} we substitute $a_s$ for $cs$. Next we factor $a_s$ out in all the brackets to have the following.
\begin{multline}
\label{eqn:prob_1set_ssparseQ2f}
\left[a_{s} a_{s}\left(\frac{1}{2} + \frac{c}{a_{s}} \right) a_{s}\left(\frac{1}{2}\right) a_{s}\left(\frac{1}{4} + \frac{c}{a_{s}}\right) a_{s}\left(\frac{1}{4}\right) a_{s}\left(\frac{1}{8} + \frac{c}{a_{s}}\right) a_{s}\left(\frac{1}{8}\right)\times \right. \\
\left.\cdots \times a_{s}\left(\frac{1}{2^{\log_2 s - 2}} + \frac{c}{a_{s}}\right) a_{s}\left(\frac{1}{2^{\log_2 s - 2}}\right)\right]^{1/2}\,.
\end{multline}

In total we have twice $(\log_2 s -2)$ plus 1 factors of $a_s$. We use this and the fact that $c/a_s = 1/s$ to simplify \eqref{eqn:prob_1set_ssparseQ2f} to \eqref{eqn:prob_1set_ssparseQ2g}, which further simplifies to \eqref{eqn:prob_1set_ssparseQ2h} by rearranging the terms in \eqref{eqn:prob_1set_ssparseQ2g}. 
\begin{multline}
\label{eqn:prob_1set_ssparseQ2g}
\left[\left(a_{s}\right)^{2\log_2 s - 3} \left(\frac{1}{2} + \frac{1}{s} \right) \left(\frac{1}{2}\right) \left(\frac{1}{4} + \frac{1}{s}\right) \left(\frac{1}{4}\right) \left(\frac{1}{8} + \frac{1}{s}\right) \left(\frac{1}{8}\right)\times \right. \\
\left.\cdots \times \left(\frac{1}{2^{\log_2 s - 2}} + \frac{1}{s}\right) \left(\frac{1}{2^{\log_2 s - 2}}\right)\right]^{1/2}\,.
\end{multline}
\begin{multline}
\label{eqn:prob_1set_ssparseQ2h}
\left[\left(a_{s}\right)^{2\log_2 s - 3}  \left(\frac{1}{2} \cdot \frac{1}{2^2} \cdot \frac{1}{2^3} \cdots \frac{1}{2^{\log_2 s - 2}}\right)\times \right. \\
\left.\cdots \times \left(\frac{1}{2} + \frac{1}{s} \right)  \left(\frac{1}{2^2} + \frac{1}{s}\right)  \left(\frac{1}{2^3} + \frac{1}{s}\right) \cdots \left(\frac{1}{2^{\log_2 s - 2}} + \frac{1}{s}\right) \right]^{1/2}\,.
\end{multline}

We focus on bounding the second line of \eqref{eqn:prob_1set_ssparseQ2h}, ignoring the square-root for the moment, that is $\left(\frac{1}{2} + \frac{1}{s} \right)  \left(\frac{1}{2^2} + \frac{1}{s}\right)  \left(\frac{1}{2^3} + \frac{1}{s}\right) \times \cdots \times \left(\frac{1}{2^{\log_2 s - 2}} + \frac{1}{s}\right)$. This equals
\begin{align}
\label{eqn:prob_1set_ssparseQ2h2}
& \exp\left(\log\left[\left(\frac{1}{2} + \frac{1}{s} \right) \left(\frac{1}{2^2} + \frac{1}{s}\right)  \left(\frac{1}{2^3} + \frac{1}{s}\right) \cdots \left(\frac{1}{2^{\log_2 s - 2}} + \frac{1}{s}\right) \right]\right)\\
\label{eqn:prob_1set_ssparseQ2h3}
& = \exp\left(\log\left[\frac{1}{2}\left(1 + \frac{2}{s} \right)\right] + \log\left[\frac{1}{2^2} \left(1 + \frac{2^2}{s}\right)\right] 
+ \cdots + \log\left[\frac{1}{2^{\log_2 s - 2}} \left(1 + \frac{2^{\log_2 s - 2}}{s}\right) \right]\right)\\
\label{eqn:prob_1set_ssparseQ2h4}
& = \exp\left(\log\left[\frac{1}{2} \cdot \frac{1}{2^2} \cdot \frac{1}{2^3} \cdots \frac{1}{2^{\log_2 s - 2}}\right] 
+ \log\left[\left(1 + \frac{2}{s} \right) \left(1 + \frac{2^2}{s}\right)
\cdots \left(1 + \frac{2^{\log_2 s - 2}}{s}\right) \right]\right)\,.
\end{align}
From \eqref{eqn:prob_1set_ssparseQ2h2} to \eqref{eqn:prob_1set_ssparseQ2h4}, we used simple algebra involving logarithms. Upper bounding $\log(1+x)$ by $x$, since $\log(1+x) \leq x$ for $|x|<1$, we upper bounded the exponent involving the second log term to upper bound \eqref{eqn:prob_1set_ssparseQ2h4} by the following.
\begin{align}
\label{eqn:prob_1set_ssparseQ2h5}
& \left(\frac{1}{2} \cdot \frac{1}{2^2} \cdot \frac{1}{2^3} \cdots \frac{1}{2^{\log_2 s - 2}}\right) 
\times \exp\left(\frac{2}{s} + \frac{2^2}{s}\ + \frac{2^3}{s} + \cdots + \frac{2^{\log_2 s - 2}}{s}\right)\\
\label{eqn:prob_1set_ssparseQ2h6}
& = \left(\frac{1}{2} \cdot \frac{1}{2^2} \cdot \frac{1}{2^3} \cdots \frac{1}{2^{\log_2 s - 2}}\right) 
\times \exp\left[\frac{1}{s}\left(\frac{s}{2} - 2\right)\right] 
\leq \left(\frac{1}{2} \cdot \frac{1}{2^2} \cdot \frac{1}{2^3} \cdots \frac{1}{2^{\log_2 s - 2}}\right) e^{\frac{1}{2}}\,.
\end{align}
The exponent of the exponential on the right of \eqref{eqn:prob_1set_ssparseQ2h5} is a geometric series and this simplifies to the LHS bound of \eqref{eqn:prob_1set_ssparseQ2h6}. The RHS bound of \eqref{eqn:prob_1set_ssparseQ2h6} is due to upper bounding $e^{1/2 - 2/s}$ by $e^{1/2}$. Using the bound in \eqref{eqn:prob_1set_ssparseQ2h6}, we upper bound \eqref{eqn:prob_1set_ssparseQ2h} by the following. 

\begin{align}
\label{eqn:prob_1set_ssparseQ2i}
\left[e^{\frac{1}{2}} \cdot \left(a_{s}\right)^{2\log_2 s - 3}  \left(\frac{1}{2} \cdot \frac{1}{2^2} \cdot \frac{1}{2^3} \cdots \frac{1}{2^{\log_2 s - 2}}\right)^2 \right]^{1/2} & = e^{\frac{1}{4}} \cdot \left(a_{s}\right)^{\log_2 s - \frac{3}{2}} \cdot \left[2^{-(1+2+\cdots + (\log_2 s-2))}\right]\\
\label{eqn:prob_1set_ssparseQ2j}
& = \frac{1}{2} e^{\frac{1}{4}} \cdot \left(a_{s}\right)^{\log_2 s - \frac{3}{2}} \cdot s^{\frac{1}{2}\log_2 s + \frac{3}{2}}\,,
\end{align}
which is the bound in \eqref{eqn:prob_1set_ssparseQ2}, hence concluding the derivation as required.

\pagebreak


\end{document}